\def\year{2021}\relax
\documentclass[letterpaper]{article} 
\usepackage{aaai21}  
\usepackage{times}  
\usepackage{helvet} 
\usepackage{courier}  
\usepackage[hyphens]{url}  
\usepackage{graphicx} 
\usepackage{natbib}  
\usepackage{caption} 
\usepackage{amsmath}
\usepackage{amsthm}
\usepackage{booktabs}
\usepackage{algorithm}
\usepackage{algorithmic}
\usepackage{amssymb}
\usepackage{mathtools}
\usepackage{tikz}
\usetikzlibrary{calc}
\usepackage{thm-restate}
\usepackage[switch]{lineno}

\newtheorem{lemma}{Lemma}

\newtheorem{definition}{Definition}
\newcommand{\bigtimes}{\text{\LARGE $\times$}}
\newcommand{\NPHard}{$\mathsf{NP}$-hard}

\newcommand{\I}{\mathcal{I}}
\newcommand{\C}{\mathcal{C}}
\DeclareMathOperator*{\argmax}{argmax}

\DeclareMathAlphabet{\mymathbb}{U}{BOONDOX-ds}{m}{n}

\title{Trembling-Hand Perfection and Correlation in Sequential Games}
\author{
    Alberto Marchesi, Nicola Gatti \\
}
\affiliations{
	Politecnico di Milano, Piazza Leonardo da Vinci 32, I-20133, Milan, Italy \\
	\texttt{\{name.surname\}@polimi.it}
}

\begin{document}


\maketitle

\begin{abstract}
	We initiate the study of trembling-hand perfection in sequential (\emph{i.e.}, extensive-form) games with correlation.
	We introduce the \emph{extensive-form perfect correlated equilibrium} (EFPCE) as a refinement of the classical \emph{extensive-form correlated equilibrium} (EFCE) that amends its weaknesses off the equilibrium path.
	This is achieved by accounting for the possibility that players may make mistakes while following recommendations independently at each information set of the game.
	After providing an axiomatic definition of EFPCE, we show that one always exists since any perfect (Nash) equilibrium constitutes an EFPCE, and that it is a refinement of EFCE, as any EFPCE is also an EFCE.
	Then, we prove that, surprisingly, computing \emph{an} EFPCE is \emph{not} harder than finding an EFCE, since the problem can be solved in polynomial time for general $n$-player extensive-form games (also with chance).
	This is achieved by formulating the problem as that of finding a limit solution (as $\epsilon \rightarrow 0$) to a suitably defined \emph{trembling} LP parametrized by $\epsilon$, featuring exponentially many variables and polynomially many constraints.
	To this end, we show how a recently developed polynomial-time algorithm for trembling LPs can be adapted to deal with problems having an exponential number of variables.
	This calls for the solution of a sequence of (non-trembling) LPs with exponentially many variables and polynomially many constraints, which is possible in polynomial time by applying an ellipsoid against hope approach.
\end{abstract}

\section{Introduction}\label{sec:intro}

\emph{Nash equilibrium} (NE)~\citep{nash1951non} computation in $2$-player zero-sum games has been the flagship challenge in artificial intelligence for several years (see, \emph{e.g.}, landmark results in poker~\citep{brown2018superhuman,brown2019superhuman}).
Recently, increasing attention has been devoted to multi-player games, where equilibria based on \emph{correlation} are now mainstream.
Correlation in games is customarily modeled through a trusted external mediator that privately recommends actions to the players.
The mediator acts as a \emph{correlation device} that draws action recommendations according to a publicly known distribution.
%
The seminal notion of \emph{correlated equilibrium} (CE) introduced by~\citet{aumann1974subjectivity} requires that no player has an incentive to deviate from a recommendation.
This is encoded by NE conditions applied to an \emph{extended} game where the correlation device plays first by randomly selecting a profile of actions according to the public distribution; then, the original game is played with each player being informed only of the action selected for her.
CEs are computationally appealing since they can be implemented in a \emph{decentralized} way by letting players play independently according to no-regret procedures~\citep{hart2000simple}.
Computing CEs in sequential (\emph{i.e.}, extensive-form) games with imperfect information has received considerable attention in the last years~\citep{celli2019learning,DBLP:conf/aaai/FarinaBS20}.
In this context, various CE definitions are possible, depending on the ways recommendations are revealed to the players.
The one that has emerged as the most suitable for sequential games is the \emph{extensive-form correlated equilibrium} (EFCE) of~\citet{von2008extensive}.
The key feature of EFCE is that recommendations are revealed to the players only when they reach a decision point where the action is to be played, and, if one player defects from a recommendation, then she stops receiving them in the future.
\citet{von2008extensive} show that EFCEs can be characterized by a polynomially-sized \emph{linear program} (LP) in two-player games without chance.
%
In the same restricted setting, \citet{DBLP:conf/nips/FarinaLFS19a} show how to find an EFCE by solving a bilinear saddle-point problem, which can be exploited to derive an efficient no-regret algorithm~\citep{DBLP:conf/nips/FarinaLFS19}.
%
%
In general $n$-player games, \citet{huang2008computing} prove that \emph{an} EFCE can be computed in polynomial time by means of an \emph{ellipsoid against hope} (EAH) algorithm similar to that introduced by~\citet{papadimitriou2008computing} for CEs in compactly represented games (see also~\citep{DBLP:conf/icml/GordonGM08} for another algorithm).
Instead, finding a payoff-maximizing EFCE is \NPHard~\citep{von2008extensive}.
Very recently,~\citet{celli2020no} provide an efficient no-regret procedure for EFCE in $n$-player games.
One of the crucial weaknesses of standard equilibrium notions, such as NE, in sequential games is that they may prescribe to play sub-optimally off the equilibrium path, \emph{i.e.}, at those information sets never reached when playing equilibrium strategies.
One way to amend this issue is \emph{trembling-hand perfection}~\citep{selten1975reexamination}, whose rationale is to let players reasoning about the possibility that they may make mistakes in the future, playing sub-optimal actions with small, vanishing probabilities (a.k.a.~trembles).
%
This idea leads to the NE refinement known as \emph{perfect equilibrium} (PE)~\citep{selten1975reexamination}.
Other refinements have been introduced in the literature; \emph{e.g.}, in the \emph{quasi-perfect equilibrium} of~\citet{van1984relation} players only account for opponents' future trembles (see~\citep{van1991stability} for other examples).
Trembles can also be introduced in normal-form games, leading to robust equilibria that rule out weakly dominated strategies~\citep{hillas2002foundations}.
Recently, equilibrium refinement has been addressed beyond the NE case, such as, \emph{e.g.}, in Stackelberg settings~\citep{farina2018trembling,marchesi2019quasi}.

Trembling-hand perfection for CEs has only been studied from a theoretical viewpoint in normal-form games, by~\citet{dhillon1996perfect}.
The authors introduce the concept of \emph{perfect} CE by enforcing PE conditions in the extended game, rather than NE ones.
Despite equilibrium refinements in sequential games are ubiquitous, no work addressed perfection and correlation together in such setting.~\footnote{Applying the perfect CE by~\citet{dhillon1996perfect} to the normal-form representation of a sequential game does \emph{not} generally solve equilibrium weaknesses. This would lead to a correlated version of the \emph{normal-form} PE, which is known not to guard against sub-optimality off the equilibrium path~\citep{van1991stability}.} 

\paragraph{Original Contributions}
We give an axiomatic definition of \emph{extensive-form perfect correlated equilibrium} (EFPCE), enforcing PE conditions, rather than NE ones, in the extended game introduced by~\citet{von2008extensive} for their original definition of EFCE.
Intuitively, this accounts for the possibility that players may make mistakes while following recommendations independently at each information set of the game.
Trembles are introduced on players' strategies, while the correlation device is defined as in classical CE notions.
First, we show that an EFPCE always exists, since any PE constitutes an EFPCE, and that EFPCE is a refinement of EFCE, as any EFPCE is also an EFCE.
Then, we show how \emph{an} EFPCE can be computed in polynomial time in any $n$-player extensive-form game (also with chance).
At first, we introduce a characterization of the equilibria of perturbed extended games (\emph{i.e.}, extended games with trembles) inspired by the definition of EFCE based on \emph{trigger agents}, introduced by~\citet{DBLP:conf/icml/GordonGM08} and~\citet{DBLP:conf/nips/FarinaLFS19a}.
This result allows us to formulate the EFPCE problem as that of finding a limit solution (as $\epsilon \rightarrow 0$) to a suitably defined \emph{trembling} LP parametrized by $\epsilon$, featuring exponentially many variables and polynomially many constraints.
To this end, we show how the polynomial-time algorithm for trembling LPs developed by~\citet{farina2018practical} can be adapted to deal with problems having an exponential number of variables.
This calls for the solution of a sequence of (non-trembling) LPs with exponentially many variables and polynomially many constraints, which is possible in polynomial time by applying an EAH approach.
The latter is inspired by the analogous algorithm of~\citet{huang2008computing} for EFCEs, which is adapted to deal with a different set of dual constraints, requiring a modification of the polynomial-time separation oracle of~\citet{huang2008computing}.~\footnote{All the omitted proofs are in Appendix~\ref{app:proofs}.}

\section{Preliminaries}\label{sec:prelim}


\begin{figure*}[!htp]\centering
	{
		\begin{minipage}[b]{4.7cm}\centering%
			\def\done{1.07*1.2}
			\def\dtwo{.53*1.2}
			\def\dleaf{.27*1.2}
			\def\dvert{.8*1.2}
			\begin{tikzpicture}[baseline=-1.7cm,scale=.95]
			
			
			\node[fill=black,draw=black,circle,inner sep=.5mm] (A) at (0, 0) {};
			\node[fill=white,draw=black,circle,inner sep=.5mm] (X) at ($(-\done,-\dvert)$) {};
			\node[fill=white,draw=black,circle,inner sep=.5mm] (Y) at ($(\done,-\dvert)$) {};
			\node[fill=black,draw=black,circle,inner sep=.5mm] (B1) at ($(X) + (-\dtwo, -\dvert)$) {};
			\node[fill=black,draw=black,circle,inner sep=.5mm] (B2) at ($(X) + (\dtwo, -\dvert)$) {};
			\node[fill=white,draw=black,inner sep=.6mm] (l1) at ($(B1) + (-\dleaf, -\dvert)$) {};
			\node[fill=white,draw=black,inner sep=.6mm] (l2) at ($(B1) + (\dleaf, -\dvert)$) {};
			\node[fill=white,draw=black,inner sep=.6mm] (l3) at ($(B2) + (-\dleaf, -\dvert)$) {};
			\node[fill=white,draw=black,inner sep=.6mm] (l4) at ($(B2) + (\dleaf, -\dvert)$) {};
			
			\node[inner sep=0] at ($(l1) + (.17,0)$) {\scriptsize$z$};
			
			\node[fill=black,draw=black,circle,inner sep=.5mm] (C) at ($(Y) + (-\dtwo, -\dvert)$) {};
			\node[fill=black,draw=black,circle,inner sep=.5mm] (D) at ($(Y) + (\dtwo, -\dvert)$) {};
			\node[fill=white,draw=black,inner sep=.6mm] (l5) at ($(C) + (-\dleaf, -\dvert)$) {};
			\node[fill=white,draw=black,inner sep=.6mm] (l7) at ($(C) + (\dleaf, -\dvert)$) {};
			\node[fill=white,draw=black,inner sep=.6mm] (l8) at ($(D) + (-\dleaf, -\dvert)$) {};
			\node[fill=white,draw=black,inner sep=.6mm] (l10) at ($(D) + (\dleaf, -\dvert)$) {};
			
			
			\draw[semithick] (A) --node[fill=white,inner sep=.9] {\scriptsize$a$} (X);
			\draw[semithick] (A) --node[fill=white,inner sep=.9] {\scriptsize$b$} (Y);
			\draw[semithick] (B1) --node[fill=white,inner sep=.9] {\scriptsize$c$} (l1);
			\draw[semithick] (B1) --node[fill=white,inner sep=.9] {\scriptsize$d$} (l2);
			\draw[semithick] (B2) --node[fill=white,inner sep=.9] {\scriptsize$c$} (l3);
			\draw[semithick] (B2) --node[fill=white,inner sep=.9] {\scriptsize$d$} (l4);
			\draw[semithick] (C) --node[fill=white,inner xsep=0,inner ysep=.9,xshift=-.4] {\scriptsize$e$} (l5);
			\draw[semithick] (C) --node[fill=white,inner xsep=.2mm,inner ysep=.9,xshift=.6] {\scriptsize$f$} (l7);
			\draw[semithick] (D) --node[fill=white,inner xsep=0,inner ysep=.9,xshift=-.4] {\scriptsize$g$} (l8);
			\draw[semithick] (D) --node[fill=white,inner xsep=.2mm,inner ysep=.9,xshift=.6] {\scriptsize$h$} (l10);
			\draw[semithick] (X) --node[fill=white,inner sep=.9] {\scriptsize$m$} (B1);
			\draw[semithick] (X) --node[fill=white,inner sep=.9] {\scriptsize$n$} (B2);
			\draw[semithick] (Y) --node[fill=white,inner sep=.9] {\scriptsize$o$} (C);
			\draw[semithick] (Y) --node[fill=white,inner sep=.9] {\scriptsize$p$} (D);
			
			
			\draw[black!60!white] (X) circle (.2);
			\node[black!60!white]  at ($(X) + (0,.38)$) {\textsc{x}};
			
			\draw[black!60!white] (Y) circle (.2);
			\node[black!60!white]  at ($(Y) + (0,.38)$) {\textsc{y}};
			
			\draw[black!60!white] (A) circle (.2);
			\node[black!60!white]  at ($(A) + (-.4, 0)$) {\textsc{i}};
			
			\draw[black!60!white] (C) circle (.2);
			\node[black!60!white]  at ($(C) + (-.2,.38)$) {\textsc{k}};
			
			\draw[black!60!white] (D) circle (.2);
			\node[black!60!white]  at ($(D) + (.2,.38)$) {\textsc{l}};
			
			\draw[black!60!white] ($(B1) + (0, .2)$) arc (90:270:.2);
			\draw[black!60!white] ($(B1) + (0, .2)$) -- ($(B2) + (0, .2)$);
			\draw[black!60!white] ($(B1) + (0, -.2)$) -- ($(B2) + (0, -.2)$);
			\draw[black!60!white] ($(B2) + (0, -.2)$) arc (-90:90:.2);
			\node[black!60!white]  at ($(B1) + (-.2, .4)$) {\textsc{j}};
			
			
			\node  at ($(l1) + (0,.-.3)$) {\scriptsize$2,2$};
			\node  at ($(l2) + (0,.-.3)$) {\scriptsize$3,1$};
			\node  at ($(l3) + (0,.-.3)$) {\scriptsize$1,3$};
			\node  at ($(l4) + (0,.-.3)$) {\scriptsize$0,0$};
			\node  at ($(l5) + (0,.-.3)$) {\scriptsize$\frac{1}{2},0$};
			\node  at ($(l7) + (0,.-.3)$) {\scriptsize$0,0$};
			\node  at ($(l8) + (0,.-.3)$) {\scriptsize$0,1$};
			\node  at ($(l10) + (0,.-.3)$) {\scriptsize$\frac{1}{2},1$};
			
			\end{tikzpicture}
	\end{minipage}}
	\hspace{.5cm}
	{\begin{minipage}[b]{4.2cm}\centering
			\setlength{\tabcolsep}{1pt}
			\begin{tabular}{c@{\hskip 6pt}cccc}
				& \textsc{i} & \textsc{j} & \textsc{k} & \textsc{l} \\[-1mm]
				\midrule\\[-5mm]
				$\pi_1$ & $a$ & $c$ & $e$ & $g$  \\[-.3mm]
				$\pi_2$ & $a$ & $c$ & $e$ & $h$  \\[-.3mm]
				$\pi_3$ & $a$ & $c$ & $f$ & $g$  \\[-.3mm]
				$\pi_4$ & $a$ & $c$ & $f$ & $h$  \\[-.3mm]
				$\pi_5$ & $a$ & $d$ & $e$ & $g$  \\[-.3mm]
				$\pi_6$ & $a$ & $d$ & $e$ & $h$  \\[-.3mm]
				$\pi_7$ & $a$ & $d$ & $f$ & $g$  \\[-.3mm]
				$\pi_8$ & $a$ & $d$ & $f$ & $h$  \\[-.3mm]
				\bottomrule
			\end{tabular}\quad
			\begin{tabular}{c@{\hskip 6pt}cccc}
				& \textsc{i} & \textsc{j} & \textsc{k} & \textsc{l} \\[-1mm]
				\midrule\\[-5mm]
				$\pi_9$    & $b$ & $c$ & $e$ & $g$  \\[-.3mm]
				$\pi_{10}$ & $b$ & $c$ & $e$ & $h$  \\[-.3mm]
				$\pi_{11}$ & $b$ & $c$ & $f$ & $g$  \\[-.3mm]
				$\pi_{12}$ & $b$ & $c$ & $f$ & $h$  \\[-.3mm]
				$\pi_{13}$ & $b$ & $d$ & $e$ & $g$  \\[-.3mm]
				$\pi_{14}$ & $b$ & $d$ & $e$ & $h$  \\[-.3mm]
				$\pi_{15}$ & $b$ & $d$ & $f$ & $g$  \\[-.3mm]
				$\pi_{16}$ & $b$ & $d$ & $f$ & $h$  \\[-.3mm]
				\bottomrule
			\end{tabular}
	\end{minipage}}
	\hspace{.5cm}
	{\begin{minipage}[b]{6.3cm}
			\setlength{\tabcolsep}{1pt}
			\begin{tabular}{rl}
				$\Pi_1(a)$ &= $ \{\pi_1, \dots, \pi_{8}\}$\\
				$\Pi_1(f)$ &= $ \{\pi_3, \pi_4, \pi_7, \pi_8, \pi_{11}, \pi_{12}, \pi_{15}, \pi_{16} \}$\\
				$\Pi_1(\textsc{i})$ &= $\{\pi_1,\dots,\pi_{16}\}$\\
				$\Pi_1(\textsc{j})$ &= $\{\pi_1, \dots, \pi_8\}$\\
				$\Pi_1(\textsc{k})$ &= $\{\pi_9, \dots, \pi_{16}\}$\\
				$\Pi_1(\textsc{l})$ &= $ \{\pi_9, \dots, \pi_{16}\}$\\
				$\Pi_1(\textsc{k}, f)$ &= $ \{\pi_3, \pi_4, \pi_7, \pi_8\}$\\
				$\Pi_1(\textsc{l}, g)$ &= $ \{\pi_9, \pi_{11}, \pi_{13}, \pi_{15}\}$\\
				$\Pi_1(z)$ &= $\{\pi_1, \pi_2, \pi_3, \pi_4\}$
			\end{tabular}%
	\end{minipage}}
	\caption{(\emph{Left}) Sample EFG. Black round nodes belong to player $1$, white round nodes belong to player $2$, and white square nodes are leaves (with players' payoffs specified under them). Rounded gray lines denote infosets. (\emph{Center}) Set $\Pi_1$ of pure strategies for player $1$. (\emph{Right}) Examples of certain subsets of $\Pi_1$ used in this work.}\label{fig:example_game}
\end{figure*}

\subsection{Extensive-Form Games}

We focus on $n$-player \emph{extensive-form games} (EFGs) with imperfect information.
We let $N \coloneqq \{ 1, \ldots, n \}$ be the set of players, and, additionally, we let $c$ be the \emph{chance} player representing exogenous stochasticity.
The sequential structure is encoded by a game tree with node set $H$.
Each node $h \in H$ is identified by the ordered sequence $\sigma(h)$ of actions encountered on the path from the root to $h$.
We let $Z \subseteq H$ be the subset of terminal nodes, which are the leaves of the game tree.
For every non-terminal node $h \in H \setminus Z$, we let $P(h) \in N \cup \{ c \}$ be the player who acts at $h$, while $A(h)$ is the set of actions available.
The function $p_c: Z \to (0,1]$ defines the probability of reaching each terminal node given the chance moves on the path from the root to that node.
For every player $i \in N$, the function $u_i : Z \to \mathbb{R}$ encodes player $i$'s utilities over terminal nodes.
Imperfect information is modeled through \emph{information sets} (infosets).
An infoset $I \subseteq H \setminus Z$ of player $i \in N$ is a group of player $i$'s nodes indistinguishable for her, \emph{i.e.}, for every $h \in I$, it must be the case that $P(h) = i$ and $A(h) = A(I)$, where $A(I)$ is the set of actions available at the infoset.
W.l.o.g., we assume that the sets $A(I)$ are disjoint.
We denote with $\I_i$ the collection of infosets of player $i \in N$.
%
For every $i \in N$, we let $A_i \coloneqq \bigcup_{I \in \I_i} A(I)$ be the set of all player $i$'s actions.
Moreover, we let $A \coloneqq \bigcup_{i \in N} A_i$.
We focus on EFGs with \emph{perfect recall} in which no player forgets what she did or knew in the past.
Formally, for every player $i \in N$ and infoset $I \in \I_i$, it must be that every node $h \in I$ is identified by the same ordered sequence $\sigma_i(I)$ of player $i$'s actions from the root to that node.
Given two infosets $I, J \in  \I_i$ of player $i \in N$, we say that $J$ \emph{follows} $I$, written $I \prec J$, if there exist two nodes $h \in I$ and $k \in J$ such that $h$ is on the path from the root to $k$.
By perfect recall, $\prec$ is a partial order on $\I_i$.
We also write $I \preceq J$ whenever either $I = J$ or $I \prec J$.
For every infoset $I \in \I_i$, we let
%
$\C(I,a) \subseteq \I_i$ be the set of all infosets that immediately follow $I$ by playing action $a \in A(I)$.

\paragraph{Strategies}
A player's \emph{pure strategy} specifies an action at every infoset of her.
For every $i \in N$, the set of player $i$'s pure strategies $\pi_i$ is $\Pi_i \coloneqq \bigtimes_{I \in \I_i} A(I)$, with $\pi_i(I) \in A(I)$ being the action at infoset $I \in \I_i$.
Moreover, $\Pi \coloneqq \bigtimes_{i \in N} \Pi_i$ denotes the set of \emph{strategy profiles} specifying a strategy for each player, while, for $i \in N$, we let $ \Pi_{-i} \coloneqq \bigtimes_{j \neq i \in N} \Pi_j$ be the (partial) strategy profiles defining a strategy for each player other than $i$.
Given $\pi_i \in \Pi_i$ and $a \in A_i$, we write $a \in \pi_i$ whenever $\pi_i$ prescribes to play $a$.
Analogously, for $\pi \in \Pi$ and $a \in A$, we write $a \in \pi$. 
%
%
%
%
Players are allowed to randomize over pure strategies by playing \emph{mixed strategies}.
For $i \in N$, we let $\mu_i: \Pi_i \to [0,1]$ be a player $i$'s mixed strategy, where $\sum_{\pi_i \in \Pi_i} \mu_i(\pi_i) = 1$.
The perfect recall assumption allows to work with \emph{behavior strategies}, which define probability distributions locally at each infoset.
For $i \in N$, we let $\beta_i : A_i \to [0,1]$ be a player $i$'s behavior strategy, which is such that $\sum_{a \in A(I)} \beta_i(a) = 1$ for all $I \in \I_i$.~\footnote{EFGs with perfect recall admit a compact strategy representation called \emph{sequence form}~\citep{von1996efficient}. See Appendix~\ref{app:sequence}.}

\paragraph{Additional Notation}
We introduce some subsets of $\Pi_i$ (see Figure~\ref{fig:example_game} for some examples).
For every action $a \in A_i$ of player $i \in N$, we define $\Pi_i(a) \coloneqq \{ \pi_i \in\Pi_i \mid a\in \pi_i  \}$ as the set of player $i$'s pure strategies specifying $a$.
For every infoset $I \in \I_i$, we let $\Pi_{i}(I) \subseteq \Pi_i$ be the set of strategies that prescribe to play so as to reach $I$ whenever possible (depending on players' moves up to that point) and \emph{any} action whenever reaching $I$ is \emph{not} possible anymore.
Additionally, for every action $a \in A(I)$, we let $\Pi_{i}(I,a) \subseteq \Pi_{i}(I) \subseteq \Pi_i$ be the set of player $i$'s strategies that reach $I$ and play $a$.
Given a terminal node $z \in Z$, we denote with $\Pi_i(z) \subseteq \Pi_i$ the set of strategies by which player $i$ plays so as to reach $z$, while $\Pi(z) \coloneqq \bigtimes_{i \in N} \Pi_i(z)$ and $\Pi_{-i}(z) \coloneqq \bigtimes_{j \neq i \in N} \Pi_j(z)$.
We also introduce the following subsets of $Z$.
For every $i \in N$ and $I \in \I_i$, we let $Z(I) \subseteq Z$ be the set of terminal nodes reachable from infoset $I$ of player $i$.
Moreover, $Z(I, a) \subseteq Z(I) \subseteq Z$ is the set of terminal nodes reachable by playing action $a \in A(I)$ at $I$, whereas
%
%
$Z^\bot(I,a) \coloneqq Z(I,a) \setminus \bigcup_{J \in \C(I,a)} Z(J)$ is the set of those reachable by playing $a$ at $I$ without traversing any other player $i$'s infoset.

\subsection{Nash Equilibrium and Its Refinements}

Given an EFG, players' behavior strategies $\{ \beta_i \}_{i \in N}$ constitute an NE if no player has an incentive to unilaterally deviate from the equilibrium by playing another strategy~\citep{nash1951non}.
%
%
%
%
The PE defined by~\citet{selten1975reexamination}
%
relies on the idea of introducing \emph{trembles} in the game, representing the possibility that players may take non-equilibrium actions with small, vanishing probability.
Trembles are encoded by means of Selten's \emph{perturbed games}, which force lower bounds on the probabilities of playing actions.
Given an EFG $\Gamma$, a pair $(\Gamma, \eta)$ defines a perturbed game, where $\eta : A \to (0,1)$ is a function assigning a positive lower bound $\eta(a)$ on the probability of playing each action $a \in A$, with $\sum_{a \in A(I)} \eta(a) < 1$ for every $i \in N$ and $I \in \I_i$.
Then:

\begin{definition}\label{def:pe}
	Given an EFG $\Gamma$, $\{ \beta_i \}_{i \in N}$ is a PE of $\Gamma$ if it is a limit point of NEs for at least one sequence of perturbed games $\{(\Gamma,\eta_t)\}_{t \in \mathbb{N}}$ such that, for all $a \in A$, the lower bounds $\eta_t(a)$ converge to zero as $t \rightarrow \infty$.
\end{definition}

%
%

%
%
%
There are only a few computational works on NE refinements.
For instance, \citet{miltersen2010computing} characterize quasi-perfect equilibria of $2$-player EFGs using the sequence form (see the recent work by~\citet{gatti2020characterization} for its extension to $n$-player games) and exploit this to compute an equilibrium by solving a linear complementarity problem with trembles defined as polynomials of some parameter treated symbolically.
%
\citet{DBLP:conf/aaai/Farina017} do the same for the PE.
%
%
%
Recently,~\citet{farina2018practical} provide a general framework for computing NE refinements in $2$-player zero-sum EFGs in polynomial time.
The authors show how to reduce the task to the more general problem of solving \emph{trembling} LPs parametrized by some parameter $\epsilon$, \emph{i.e.}, finding their limit solutions as $\epsilon \rightarrow 0$.
Then, they provide a general polynomial-time algorithm to find limit solutions to trembling LPs.
Other works study the problem of computing (approximate) NE refinements in $2$-player zero-sum EFGs by employing online convex optimization techniques~\citep{DBLP:conf/ijcai/KroerFS17,farina2017regret}.

\subsection{Correlation in Extensive-Form Games}

%
%
We model a correlation device as a probability distribution $\mu \in \Delta_{\Pi}$.
%
%
%
In the classical CE by~\citet{aumann1974subjectivity}, the correlation device draws a strategy profile $\pi \in \Pi$ according to $\mu$; then, it privately communicates $\pi_i$ to each player $i \in N$.
%
%
This notion of CE does \emph{not} fit well to EFGs, as it requires the players to reason over the exponentially-sized set $\Pi_i$.
%
%
%
\citet{von2008extensive} introduced the EFCE to solve this issue.
The first crucial feature of the EFCE is a different way of giving recommendations: the strategy $\pi_i$ is revealed to player $i$ as the game progresses, \emph{i.e.}, the player is recommended to play the action $\pi_i(I)$ at infoset $I \in \I_i$ only when $I$ is actually reached during play.
The second key aspect characterizing EFCEs is that, whenever a player decides to defect from a recommended action at some infoset, then she may choose any move at her subsequent infosets and she stops receiving recommendations from the correlation device.
%
%
The definition of EFCE introduced by~\citet{von2008extensive} (Definition~\ref{def:efce}) requires the introduction of the notion of \emph{extended game} with a correlation device.

\begin{definition}\label{def:ext_game}
	Given an EFG $\Gamma$ and a distribution $\mu \in \Delta_{\Pi}$, the extended game $\Gamma^{\textnormal{ext}}(\mu)$ is a new EFG in which chance first selects $\pi \in \Pi$ according to $\mu$, and, then, $\Gamma$ is played with each player $i \in N$ receiving the recommendation to play $\pi_i(I)$ as a signal, whenever she reaches an infoset $I \in \I_i$.
\end{definition}

The signaling in $\Gamma^{\textnormal{ext}}(\mu)$ induces a new infoset structure.
Specifically, every infoset $I \in \I_i$ of the original game $\Gamma$ corresponds to many, new infosets in $\Gamma^{\textnormal{ext}}(\mu)$, one for each combination of possible action recommendations received at the infosets preceding $I$ (this included).
At each new infoset, player $i$ can only distinguish among chance moves corresponding to strategy profiles $\pi \in \Pi$ that differ in the recommendations at infosets $J \in \I_i: J \preceq I$.
%
Figure~\ref{fig:extended_game} shows a simple EFG with its corresponding extended game.

\begin{figure*}[!htp]\centering
	{
		\begin{minipage}[b]{2cm}\centering%
			\def\done{.55*1.2}
			\def\dtwo{.45*1.2}
			\def\dleaf{.2*1.2}
			\def\dvert{.9*1.2}
			\begin{tikzpicture}[baseline=-1.1cm,scale=.95]
			\node[fill=black,draw=black,circle,inner sep=.5mm] (A) at (0, 0) {};
			\node[fill=white,draw=black,circle,inner sep=.5mm] (X) at ($(-\done,-\dvert)$) {};
			\node[fill=white,draw=black,inner sep=.6mm] (l5) at ($(\done,-\dvert)$) {};
			\node[fill=black,draw=black,circle,inner sep=.5mm] (B1) at ($(X) + (-\dtwo, -\dvert)$) {};
			\node[fill=black,draw=black,circle,inner sep=.5mm] (B2) at ($(X) + (\dtwo, -\dvert)$) {};
			\node[fill=white,draw=black,inner sep=.6mm] (l1) at ($(B1) + (-\dleaf, -\dvert)$) {};
			\node[fill=white,draw=black,inner sep=.6mm] (l2) at ($(B1) + (\dleaf, -\dvert)$) {};
			\node[fill=white,draw=black,inner sep=.6mm] (l3) at ($(B2) + (-\dleaf, -\dvert)$) {};
			\node[fill=white,draw=black,inner sep=.6mm] (l4) at ($(B2) + (\dleaf, -\dvert)$) {};
			
			\draw[semithick] (A) --node[fill=white,inner sep=.9] {\scriptsize$a$} (X);
			\draw[semithick] (A) --node[fill=white,inner sep=.9] {\scriptsize$b$} (l5);
			\draw[semithick] (B1) --node[fill=white,inner sep=.9] {\scriptsize$c$} (l1);
			\draw[semithick] (B1) --node[fill=white,inner sep=.9] {\scriptsize$d$} (l2);
			\draw[semithick] (B2) --node[fill=white,inner sep=.9] {\scriptsize$c$} (l3);
			\draw[semithick] (B2) --node[fill=white,inner sep=.9] {\scriptsize$d$} (l4);
			\draw[semithick] (X) --node[fill=white,inner sep=.9] {\scriptsize$m$} (B1);
			\draw[semithick] (X) --node[fill=white,inner sep=.9] {\scriptsize$n$} (B2);
			
			\draw[black!60!white] (X) circle (.2);
			\node[black!60!white]  at ($(X) + (-.2,.38)$) {\textsc{x}};
			
			\draw[black!60!white] (A) circle (.2);
			\node[black!60!white]  at ($(A) + (-.4, 0)$) {\textsc{i}};
			
			\draw[black!60!white] ($(B1) + (0, .2)$) arc (90:270:.2);
			\draw[black!60!white] ($(B1) + (0, .2)$) -- ($(B2) + (0, .2)$);
			\draw[black!60!white] ($(B1) + (0, -.2)$) -- ($(B2) + (0, -.2)$);
			\draw[black!60!white] ($(B2) + (0, -.2)$) arc (-90:90:.2);
			\node[black!60!white]  at ($(B1) + (-.2, .4)$) {\textsc{j}};
			
			\end{tikzpicture}
	\end{minipage}}
	\hspace{.5cm}
	{\begin{minipage}[b]{13.5cm}\centering
			\def\dchance{1.4*1.2}
			\def\dchancehalf{.7*1.2}
			\def\dvertchance{1.1*1.2}
			\def\done{.3*1.2}
			\def\dtwo{.3*1.2}
			\def\dleaf{.15*1.2}
			\def\dvert{.5*1.2}
			\def\dverthigh{0.9*1.2}
			\def\dvertlow{.25*1.2}
			\begin{tikzpicture}[baseline=-1.1cm,scale=.95]
			\node[fill=white,draw=black,inner sep=1mm] (Ch) at (0, 0) {};
			
			\node[fill=black,draw=black,circle,inner sep=.5mm] (A) at ($(Ch) + (-\dchance -\dchance -\dchance -\dchancehalf, -\dvertchance)$) {};
			\node[fill=white,draw=black,circle,inner sep=.5mm] (X) at ($(A) + (-\done,-\dvert)$) {};
			\node[fill=white,draw=black,inner sep=.5mm] (l5) at ($(A) + (\done,-\dvertlow)$) {};
			\node[fill=black,draw=black,circle,inner sep=.5mm] (B1) at ($(X) + (-\dtwo, -\dverthigh)$) {};
			\node[fill=black,draw=black,circle,inner sep=.5mm] (B2) at ($(X) + (\dtwo, -\dverthigh)$) {};
			\node[fill=white,draw=black,inner sep=.5mm] (l1) at ($(B1) + (-\dleaf, -\dvert)$) {};
			\node[fill=white,draw=black,inner sep=.5mm] (l2) at ($(B1) + (\dleaf, -\dvert)$) {};
			\node[fill=white,draw=black,inner sep=.5mm] (l3) at ($(B2) + (-\dleaf, -\dvert)$) {};
			\node[fill=white,draw=black,inner sep=.5mm] (l4) at ($(B2) + (\dleaf, -\dvert)$) {};
			\draw[semithick] (Ch) --node[fill=white,sloped,inner sep=.2] {\scriptsize$\mu(ac,m)$} (A);
			\draw[very thick] (A) -- (X);
			\draw[semithick] (A) -- (l5);
			\draw[very thick] (B1) -- (l1);
			\draw[semithick] (B1) -- (l2);
			\draw[very thick] (B2) -- (l3);
			\draw[semithick] (B2) -- (l4);
			\draw[very thick] (X) -- (B1);
			\draw[semithick] (X) -- (B2);
			
			\node[fill=black,draw=black,circle,inner sep=.5mm] (A2) at ($(Ch) + (-\dchance -\dchance -\dchancehalf, -\dvertchance)$) {};
			\node[fill=white,draw=black,circle,inner sep=.5mm] (X2) at ($(A2) + (-\done,-\dverthigh)$) {};
			\node[fill=white,draw=black,inner sep=.5mm] (l52) at ($(A2) + (\done,-\dvertlow)$) {};
			\node[fill=black,draw=black,circle,inner sep=.5mm] (B12) at ($(X2) + (-\dtwo, -\dvert)$) {};
			\node[fill=black,draw=black,circle,inner sep=.5mm] (B22) at ($(X2) + (\dtwo, -\dvert)$) {};
			\node[fill=white,draw=black,inner sep=.5mm] (l12) at ($(B12) + (-\dleaf, -\dvert)$) {};
			\node[fill=white,draw=black,inner sep=.5mm] (l22) at ($(B12) + (\dleaf, -\dvert)$) {};
			\node[fill=white,draw=black,inner sep=.5mm] (l32) at ($(B22) + (-\dleaf, -\dvert)$) {};
			\node[fill=white,draw=black,inner sep=.5mm] (l42) at ($(B22) + (\dleaf, -\dvert)$) {};
			\draw[semithick] (Ch) --node[fill=white,sloped,inner sep=.2] {\scriptsize$\mu(ac,n)$} (A2);
			\draw[very thick] (A2) -- (X2);
			\draw[semithick] (A2) -- (l52);
			\draw[very thick] (B12) -- (l12);
			\draw[semithick] (B12) -- (l22);
			\draw[very thick] (B22) -- (l32);
			\draw[semithick] (B22) -- (l42);
			\draw[semithick] (X2) -- (B12);
			\draw[very thick] (X2) -- (B22);
			
			\node[fill=black,draw=black,circle,inner sep=.5mm] (A3) at ($(Ch) + (-\dchance  -\dchancehalf, -\dvertchance)$) {};
			\node[fill=white,draw=black,circle,inner sep=.5mm] (X3) at ($(A3) + (-\done,-\dvert)$) {};
			\node[fill=white,draw=black,inner sep=.5mm] (l53) at ($(A3) + (\done,-\dvertlow)$) {};
			\node[fill=black,draw=black,circle,inner sep=.5mm] (B13) at ($(X3) + (-\dtwo, -\dverthigh)$) {};
			\node[fill=black,draw=black,circle,inner sep=.5mm] (B23) at ($(X3) + (\dtwo, -\dverthigh)$) {};
			\node[fill=white,draw=black,inner sep=.5mm] (l13) at ($(B13) + (-\dleaf, -\dvert)$) {};
			\node[fill=white,draw=black,inner sep=.5mm] (l23) at ($(B13) + (\dleaf, -\dvert)$) {};
			\node[fill=white,draw=black,inner sep=.5mm] (l33) at ($(B23) + (-\dleaf, -\dvert)$) {};
			\node[fill=white,draw=black,inner sep=.5mm] (l43) at ($(B23) + (\dleaf, -\dvert)$) {};
			\draw[semithick] (Ch) --node[fill=white,sloped,inner sep=.2] {\scriptsize$\mu(ad,m)$} (A3);
			\draw[very thick] (A3) -- (X3);
			\draw[semithick] (A3) -- (l53);
			\draw[semithick] (B13) -- (l13);
			\draw[very thick] (B13) -- (l23);
			\draw[semithick] (B23) -- (l33);
			\draw[very thick] (B23) -- (l43);
			\draw[very thick] (X3) -- (B13);
			\draw[semithick] (X3) -- (B23);
			
			\node[fill=black,draw=black,circle,inner sep=.5mm] (A4) at ($(Ch) + (-\dchancehalf, -\dvertchance)$) {};
			\node[fill=white,draw=black,circle,inner sep=.5mm] (X4) at ($(A4) + (-\done,-\dverthigh)$) {};
			\node[fill=white,draw=black,inner sep=.5mm] (l54) at ($(A4) + (\done,-\dvertlow)$) {};
			\node[fill=black,draw=black,circle,inner sep=.5mm] (B14) at ($(X4) + (-\dtwo, -\dvert)$) {};
			\node[fill=black,draw=black,circle,inner sep=.5mm] (B24) at ($(X4) + (\dtwo, -\dvert)$) {};
			\node[fill=white,draw=black,inner sep=.5mm] (l14) at ($(B14) + (-\dleaf, -\dvert)$) {};
			\node[fill=white,draw=black,inner sep=.5mm] (l24) at ($(B14) + (\dleaf, -\dvert)$) {};
			\node[fill=white,draw=black,inner sep=.5mm] (l34) at ($(B24) + (-\dleaf, -\dvert)$) {};
			\node[fill=white,draw=black,inner sep=.5mm] (l44) at ($(B24) + (\dleaf, -\dvert)$) {};
			\draw[semithick] (Ch) --node[fill=white,sloped,inner sep=.2] {\scriptsize$\mu(ad,n)$} (A4);
			\draw[very thick] (A4) -- (X4);
			\draw[semithick] (A4) -- (l54);
			\draw[semithick] (B14) -- (l14);
			\draw[very thick] (B14) -- (l24);
			\draw[semithick] (B24) -- (l34);
			\draw[very thick] (B24) -- (l44);
			\draw[semithick] (X4) -- (B14);
			\draw[very thick] (X4) -- (B24);
			
			\node[fill=black,draw=black,circle,inner sep=.5mm] (A5) at ($(Ch) + (\dchancehalf, -\dvertchance)$) {};
			\node[fill=white,draw=black,circle,inner sep=.5mm] (X5) at ($(A5) + (-\done,-\dvert)$) {};
			\node[fill=white,draw=black,inner sep=.5mm] (l55) at ($(A5) + (\done,-\dvertlow)$) {};
			\node[fill=black,draw=black,circle,inner sep=.5mm] (B15) at ($(X5) + (-\dtwo, -\dverthigh)$) {};
			\node[fill=black,draw=black,circle,inner sep=.5mm] (B25) at ($(X5) + (\dtwo, -\dverthigh)$) {};
			\node[fill=white,draw=black,inner sep=.5mm] (l15) at ($(B15) + (-\dleaf, -\dvert)$) {};
			\node[fill=white,draw=black,inner sep=.5mm] (l25) at ($(B15) + (\dleaf, -\dvert)$) {};
			\node[fill=white,draw=black,inner sep=.5mm] (l35) at ($(B25) + (-\dleaf, -\dvert)$) {};
			\node[fill=white,draw=black,inner sep=.5mm] (l45) at ($(B25) + (\dleaf, -\dvert)$) {};
			\draw[semithick] (Ch) --node[fill=white,sloped,inner sep=.2] {\scriptsize$\mu(bc,m)$} (A5);
			\draw[semithick] (A5) -- (X5);
			\draw[very thick] (A5) -- (l55);
			\draw[very thick] (B15) -- (l15);
			\draw[semithick] (B15) -- (l25);
			\draw[very thick] (B25) -- (l35);
			\draw[semithick] (B25) -- (l45);
			\draw[very thick] (X5) -- (B15);
			\draw[semithick] (X5) -- (B25);
			
			\node[fill=black,draw=black,circle,inner sep=.5mm] (A6) at ($(Ch) + (\dchancehalf + \dchance, -\dvertchance)$) {};
			\node[fill=white,draw=black,circle,inner sep=.5mm] (X6) at ($(A6) + (-\done,-\dverthigh)$) {};
			\node[fill=white,draw=black,inner sep=.5mm] (l56) at ($(A6) + (\done,-\dvertlow)$) {};
			\node[fill=black,draw=black,circle,inner sep=.5mm] (B16) at ($(X6) + (-\dtwo, -\dvert)$) {};
			\node[fill=black,draw=black,circle,inner sep=.5mm] (B26) at ($(X6) + (\dtwo, -\dvert)$) {};
			\node[fill=white,draw=black,inner sep=.5mm] (l16) at ($(B16) + (-\dleaf, -\dvert)$) {};
			\node[fill=white,draw=black,inner sep=.5mm] (l26) at ($(B16) + (\dleaf, -\dvert)$) {};
			\node[fill=white,draw=black,inner sep=.5mm] (l36) at ($(B26) + (-\dleaf, -\dvert)$) {};
			\node[fill=white,draw=black,inner sep=.5mm] (l46) at ($(B26) + (\dleaf, -\dvert)$) {};
			\draw[semithick] (Ch) --node[fill=white,sloped,inner sep=.2] {\scriptsize$\mu(bc,n)$} (A6);
			\draw[semithick] (A6) -- (X6);
			\draw[very thick] (A6) -- (l56);
			\draw[very thick] (B16) -- (l16);
			\draw[semithick] (B16) -- (l26);
			\draw[very thick] (B26) -- (l36);
			\draw[semithick] (B26) -- (l46);
			\draw[semithick] (X6) -- (B16);
			\draw[very thick] (X6) -- (B26);
			
			\node[fill=black,draw=black,circle,inner sep=.5mm] (A7) at ($(Ch) + (\dchancehalf + \dchance + \dchance, -\dvertchance)$) {};
			\node[fill=white,draw=black,circle,inner sep=.5mm] (X7) at ($(A7) + (-\done,-\dvert)$) {};
			\node[fill=white,draw=black,inner sep=.5mm] (l57) at ($(A7) + (\done,-\dvertlow)$) {};
			\node[fill=black,draw=black,circle,inner sep=.5mm] (B17) at ($(X7) + (-\dtwo, -\dverthigh)$) {};
			\node[fill=black,draw=black,circle,inner sep=.5mm] (B27) at ($(X7) + (\dtwo, -\dverthigh)$) {};
			\node[fill=white,draw=black,inner sep=.5mm] (l17) at ($(B17) + (-\dleaf, -\dvert)$) {};
			\node[fill=white,draw=black,inner sep=.5mm] (l27) at ($(B17) + (\dleaf, -\dvert)$) {};
			\node[fill=white,draw=black,inner sep=.5mm] (l37) at ($(B27) + (-\dleaf, -\dvert)$) {};
			\node[fill=white,draw=black,inner sep=.5mm] (l47) at ($(B27) + (\dleaf, -\dvert)$) {};
			\draw[semithick] (Ch) --node[fill=white,sloped,inner sep=.2] {\scriptsize$\mu(bd,m)$} (A7);
			\draw[semithick] (A7) -- (X7);
			\draw[very thick] (A7) -- (l57);
			\draw[semithick] (B17) -- (l17);
			\draw[very thick] (B17) -- (l27);
			\draw[semithick] (B27) -- (l37);
			\draw[very thick] (B27) -- (l47);
			\draw[very thick] (X7) -- (B17);
			\draw[semithick] (X7) -- (B27);
			
			\node[fill=black,draw=black,circle,inner sep=.5mm] (A8) at ($(Ch) + (\dchancehalf + \dchance + \dchance + \dchance, -\dvertchance)$) {};
			\node[fill=white,draw=black,circle,inner sep=.5mm] (X8) at ($(A8) + (-\done,-\dverthigh)$) {};
			\node[fill=white,draw=black,inner sep=.5mm] (l58) at ($(A8) + (\done,-\dvertlow)$) {};
			\node[fill=black,draw=black,circle,inner sep=.5mm] (B18) at ($(X8) + (-\dtwo, -\dvert)$) {};
			\node[fill=black,draw=black,circle,inner sep=.5mm] (B28) at ($(X8) + (\dtwo, -\dvert)$) {};
			\node[fill=white,draw=black,inner sep=.5mm] (l18) at ($(B18) + (-\dleaf, -\dvert)$) {};
			\node[fill=white,draw=black,inner sep=.5mm] (l28) at ($(B18) + (\dleaf, -\dvert)$) {};
			\node[fill=white,draw=black,inner sep=.5mm] (l38) at ($(B28) + (-\dleaf, -\dvert)$) {};
			\node[fill=white,draw=black,inner sep=.5mm] (l48) at ($(B28) + (\dleaf, -\dvert)$) {};
			\draw[semithick] (Ch) --node[fill=white,sloped,inner sep=.2] {\scriptsize$\mu(bd,n)$} (A8);
			\draw[semithick] (A8) -- (X8);
			\draw[very thick] (A8) -- (l58);
			\draw[semithick] (B18) -- (l18);
			\draw[very thick] (B18) -- (l28);
			\draw[semithick] (B28) -- (l38);
			\draw[very thick] (B28) -- (l48);
			\draw[semithick] (X8) -- (B18);
			\draw[very thick] (X8) -- (B28);

			\draw[black!60!white] ($(A) + (0, .2)$) arc (90:270:.2);
			\draw[black!60!white] ($(A) + (0, .2)$) -- ($(A4) + (0, .2)$);
			\draw[black!60!white] ($(A) + (0, -.2)$) -- ($(A4) + (0, -.2)$);
			\draw[black!60!white] ($(A4) + (0, -.2)$) arc (-90:90:.2);
			\node[black!60!white]  at ($(A) + (-.5, .3)$) {\scriptsize$\langle \textsc{i},a \rangle$};
			\draw[black!60!white] ($(A5) + (0, .2)$) arc (90:270:.2);
			\draw[black!60!white] ($(A5) + (0, .2)$) -- ($(A8) + (0, .2)$);
			\draw[black!60!white] ($(A5) + (0, -.2)$) -- ($(A8) + (0, -.2)$);
			\draw[black!60!white] ($(A8) + (0, -.2)$) arc (-90:90:.2);
			\node[black!60!white]  at ($(A8) + (.5, .3)$) {\scriptsize$\langle \textsc{i},b \rangle$};
			
			\draw[black!60!white] ($(X) + (0, .2)$) arc (90:270:.2);
			\draw[black!60!white] ($(X) + (0, .2)$) -- ($(X7) + (0, .2)$);
			\draw[black!60!white] ($(X) + (0, -.2)$) -- ($(X7) + (0, -.2)$);
			\draw[black!60!white] ($(X7) + (0, -.2)$) arc (-90:90:.2);
			\node[black!60!white]  at ($(X) + (-.5, .35)$) {\scriptsize$\langle \textsc{x},m \rangle$};
			\draw[black!60!white] ($(X2) + (0, .2)$) arc (90:270:.2);
			\draw[black!60!white] ($(X2) + (0, .2)$) -- ($(X8) + (0, .2)$);
			\draw[black!60!white] ($(X2) + (0, -.2)$) -- ($(X8) + (0, -.2)$);
			\draw[black!60!white] ($(X8) + (0, -.2)$) arc (-90:90:.2);
			\node[black!60!white]  at ($(X8) + (.6, .3)$) {\scriptsize$\langle \textsc{x},n \rangle$};
			
			\draw[black!60!white] ($(B1) + (0, .2)$) arc (90:270:.2);
			\draw[black!60!white] ($(B1) + (0, .2)$) -- ($(B22) + (0, .2)$);
			\draw[black!60!white] ($(B1) + (0, -.2)$) -- ($(B22) + (0, -.2)$);
			\draw[black!60!white] ($(B22) + (0, -.2)$) arc (-90:90:.2);
			\node[black!60!white,fill=white,inner sep=.2]  at ($(B12) + (-.5, .2)$) {\scriptsize$\langle \textsc{j},ac \rangle$};
			\draw[black!60!white] ($(B13) + (0, .2)$) arc (90:270:.2);
			\draw[black!60!white] ($(B13) + (0, .2)$) -- ($(B24) + (0, .2)$);
			\draw[black!60!white] ($(B13) + (0, -.2)$) -- ($(B24) + (0, -.2)$);
			\draw[black!60!white] ($(B24) + (0, -.2)$) arc (-90:90:.2);
			\node[black!60!white,fill=white,inner sep=.2]  at ($(B14) + (-.5, .2)$) {\scriptsize$\langle \textsc{j},ad \rangle$};
			\draw[black!60!white] ($(B15) + (0, .2)$) arc (90:270:.2);
			\draw[black!60!white] ($(B15) + (0, .2)$) -- ($(B26) + (0, .2)$);
			\draw[black!60!white] ($(B15) + (0, -.2)$) -- ($(B26) + (0, -.2)$);
			\draw[black!60!white] ($(B26) + (0, -.2)$) arc (-90:90:.2);
			\node[black!60!white,fill=white,inner sep=.2]  at ($(B16) + (-.5, .2)$) {\scriptsize$\langle \textsc{j},bc \rangle$};
			\draw[black!60!white] ($(B17) + (0, .2)$) arc (90:270:.2);
			\draw[black!60!white] ($(B17) + (0, .2)$) -- ($(B28) + (0, .2)$);
			\draw[black!60!white] ($(B17) + (0, -.2)$) -- ($(B28) + (0, -.2)$);
			\draw[black!60!white] ($(B28) + (0, -.2)$) arc (-90:90:.2);
			\node[black!60!white,fill=white,inner sep=.2]  at ($(B18) + (-.5, .2)$) {\scriptsize$\langle \textsc{j},bd \rangle$};
			
			\end{tikzpicture}
	\end{minipage}}
	\caption{(\emph{Left}) An EFG $\Gamma$. (\emph{Right}) The extended game $\Gamma^\textnormal{ext}(\mu)$. The white square node at the root is a chance node, where each action corresponds to some $\pi \in \Pi$ and is labeled with its probability $\mu(\pi)$. Infosets in $\Gamma^\textnormal{ext}(\mu)$ are identified by pairs. For instance, infoset $\langle \textsc{j},ac \rangle$ corresponds to $\textsc{j}$ when being recommended to play $a$ and $c$ at $\textsc{i}$ and $\textsc{j}$, respectively. Thick actions represent players' behavior when following recommendations (for the ease of reading, action names are omitted).}\label{fig:extended_game}
\end{figure*}
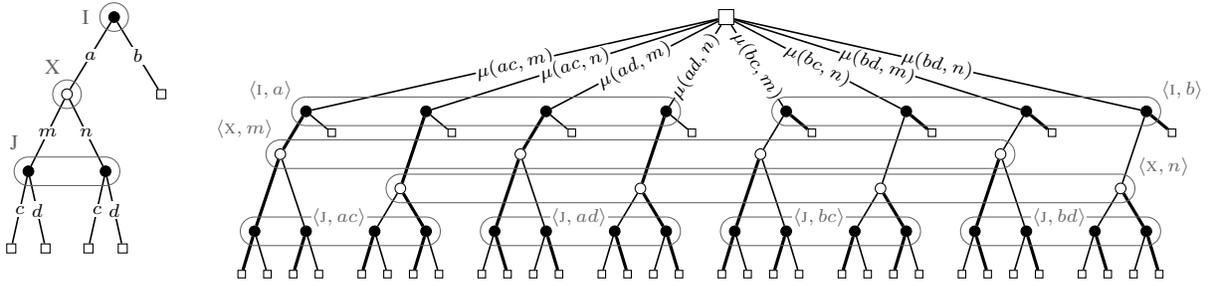

\begin{definition}\label{def:efce}
	Given an EFG $\Gamma$, $\mu \in \Delta_{\Pi}$ defines an EFCE of $\Gamma$ if following recommendations is an NE of $\Gamma^{\textnormal{ext}}(\mu)$.~\footnote{For EFCEs, one can restrict the attention to distributions $\mu$ over \emph{reduced} strategy profiles, \emph{i.e.}, those in which each player's pure strategy only specifies actions at infosets reachable given that player's moves~\citep{vermeulen1998reduced}. In the following, we stick to general, un-reduced strategy profiles since, as showed in Appendix~\ref{app:reduced}, these are necessary for trembling-hand perfect CEs in order to define the players' behavior off the equilibrium path.}
\end{definition}

Next, we introduce an equivalent characterization of EFCEs~\citep{DBLP:conf/aaai/FarinaBS20}.
It is based on the following concept of \emph{trigger agent}, originally due to~\citet{DBLP:conf/icml/GordonGM08}.

\begin{definition}\label{def:trigger}
	Given an infoset $I \in \I_i$ of player $i \in N$, an action $a \in A(I)$, and a distribution $\hat \mu_i \in \Delta_{\Pi_i(I)}$, an \emph{$(I, a, \hat \mu_i)$-trigger agent for player $i$} is an agent that takes on the role of player $i$ and follows all recommendations unless she reaches $I$ and gets recommended to play $a$.
	If this happens, she stops committing to recommendations and plays according to a strategy sampled from $\hat \mu_i$ until the game ends.
\end{definition}

%
Then, it follows that $\mu \in \Delta_\Pi$ is an EFCE if, for every $i \in N$, player $i$'s expected utility when following recommendations is at least as large as the expected utility that any $(I, a, \hat \mu_i)$-trigger agent for player $i$ can achieve (assuming the opponents' do not deviate from recommendations).
We provide a formal statement in Appendix~\ref{app:trigger}.

\paragraph{Computing EFCEs in $n$-player EFGs}
%
%
The algorithm of~\citet{huang2008computing} relies on the following LP formulation of the problem of finding an EFCE, which has exponentially many variables and polynomially many constraints (for completeness, its derivation is in Appendix~\ref{app:lp}).
\begin{subequations}\label{prob:primal_ellip}
	\begin{align}
		\max_{\boldsymbol{\mu} \geq \boldsymbol{0}, \boldsymbol{v}} & \quad \sum_{\pi \in \Pi} \mu[\pi] \quad \textnormal{s.t.} \\
		& A \boldsymbol{\mu} + B \boldsymbol{v }\geq \boldsymbol{0},
	\end{align}
\end{subequations}
where $\boldsymbol{\mu}$ is a vector of variables $\mu[\pi]$ for $\pi \in \Pi$, encoding a probability distribution $\mu \in \Delta_{\Pi}$.
Problem~\ref{prob:primal_ellip} does not enforce any simplex constraint on variables $\mu[\pi]$, and, thus, it is either unbounded or it has an optimal solution with value zero (by setting $\boldsymbol{\mu}$ and $\boldsymbol{v}$ to zero).
In the former case, any feasible $\boldsymbol{\mu}$ encodes an EFCE after normalizing it.
%
As a result, since an EFCE always exists~\citep{von2008extensive}, the following dual of Problem~\ref{prob:primal_ellip} is always infeasible:
\begin{subequations}\label{prob:dual_ellip}
	\begin{eqnarray}
	A^\top \boldsymbol{y} &\leq& -\boldsymbol{1} \\
	B^\top \boldsymbol{y} &=& \textcolor{white}{-}\boldsymbol{0} \\
	\boldsymbol{y} &\geq& \textcolor{white}{-}\boldsymbol{0},
	\end{eqnarray}
\end{subequations}
where $\boldsymbol{y}$ is a vector of dual variables.
The EAH approach applies the ellipsoid algorithm~\citep{grotschel1993geometric} to Problem~\ref{prob:dual_ellip} in order to conclude that it is infeasible.
Since there are exponentially many constraints, the algorithm runs in polynomial time only if a polynomial-time separation oracle is available.
This is given by the following: 
\begin{lemma}[Lemma~5, \citep{huang2008computing}]\label{lem:sep_ro}
	If $\boldsymbol{y} \geq \boldsymbol{0}$ is such that $B^\top \boldsymbol{y} = \boldsymbol{0}$, then there exists $\boldsymbol{\mu}$ encoding a product distribution $\mu \in \Delta_{\Pi}$ such that $\boldsymbol{\mu}^\top A^\top \boldsymbol{y} = 0$.
	Moreover, $\boldsymbol{\mu}$ can be computed in polynomial time.
\end{lemma}
\noindent
\citet{jiang2015polynomial} show how, given a product distribution $\mu$ computed as in Lemma~\ref{lem:sep_ro}, it is possible to recover, in polynomial time, a violated constraint for Problem~\ref{prob:dual_ellip}, corresponding to some strategy profile $\pi \in \Pi$.
%
%
%
This, together with some additional technical tricks ensuring that $B^\top \boldsymbol{y} = \boldsymbol{0}$ holds (see~\citep{huang2008computing} for more details), allows to apply the ellipsoid algorithm to Problem~\ref{prob:dual_ellip} in polynomial time.
Since the problem is infeasible, the algorithm must terminate after polynomially many iterations with a collection of violated constraints, which correspond to polynomially many strategy profiles.
Then, solving (in polynomial time) Problem~\ref{prob:primal_ellip} with the variables $\boldsymbol{\mu}$ restricted to these strategy profiles gives an EFCE of the game. 
Let us also remark that the EFCE obtained in this way has support size polynomial in the size of the game.

\section{Trembling-Hand Perfection and Correlation}\label{sec:perf_corr}

We are now ready to show how trembling-hand perfection can be injected into the definition of EFCE so as to amend its weaknesses off the equilibrium path (see the following for an example).
We generalize the approach of~\citet{dhillon1996perfect} (restricted to CEs in normal-form games) to the general setting of EFCEs in EFGs.
The core idea is to use the PE rather than the NE in the definition of CE.
Thus:
%

\begin{definition}\label{def:efpce}
	Given an EFG $\Gamma$, a distribution $\mu \in \Delta_{\Pi}$ is an \emph{extensive-form perfect correlated equilibrium (EFPCE)} if following recommendations is a PE of $\Gamma^{\textnormal{ext}}(\mu)$.
\end{definition}

The definition of EFPCE crucially relies on the introduction of trembles in extended games, \emph{i.e.}, it takes into account the possibility that each player may not follow action recommendations with a small, vanishing probability.
In the following, given a perturbed EFG $(\Gamma, \eta)$ and $\mu \in \Delta_\Pi$, we denote with $(\Gamma^{\textnormal{ext}}(\mu), \eta)$ a perturbed extended game in which the probability of playing each action is subject to a lower bound equal to the lower bound $\eta(a)$ of the corresponding action $a \in A$ in $\Gamma$.
%
%
%
By recalling the definition of PE (Definition~\ref{def:pe}) and the structure of perturbed extended games, it is easy to infer the following characterization of EFPCEs:

\begin{lemma}\label{lem:efpce_limit_point}
	Given an EFG $\Gamma$, a distribution $\mu \in \Delta_{\Pi}$ is an EFPCE of $\Gamma$ if following recommendations constitutes NEs for at least one sequence of perturbed extended games $\{ \Gamma^{\textnormal{ext}}(\mu), \eta_t) \}_{t \in \mathbb{N}}$ such that, for all $a \in A$, the lower bounds $\eta_t(a)$ converge to zero as $t \rightarrow \infty $.
\end{lemma}

We remark that, with an abuse of terminology, we say that players follow recommendations in a perturbed extended game $(\Gamma^{\textnormal{ext}}(\mu), \eta)$ whenever they play strategies which place all the residual probability (given lower bounds) on recommended actions.
In the following sections, we crucially rely on the characterization of EFPCEs given in Lemma~\ref{lem:efpce_limit_point} in order to derive our computational results.
First, we show an example of EFPCE and prove some of its properties.

\paragraph{Example of EFPCE}
Consider the EFG in Figure~\ref{fig:example_game}(\emph{Left}) and lower bounds $\eta_t: A \to (0,1)$ for $t \in \mathbb{N}$, with $\eta_t(a) \rightarrow 0$ as $t \rightarrow \infty$ for all $a \in A$.
First, notice that player $1$ is always better off playing action $a$ at the root infoset $\textsc{i}$, since she can guarantee herself a utility of $1$ by selecting $c$ at the following infoset $\textsc{j}$, while she can achieve at most $\frac{1}{2}$ by playing $b$.
Thus, any EFPCE of the game (as well as any EFCE) must recommend $a$ at $\textsc{i}$ with probability $1$.
%
Then, in the sub-game reached when playing $a$ at $\textsc{i}$, it is easy to check that recommending the pairs of actions $(c,m)$, $(c,n)$, and $(d,m)$ each with probability $\frac{1}{3}$ is an equilibrium, as no player has an incentive to deviate from a recommendation, even with trembles (see Appendix~\ref{app:example} for more details).
The correlation device described fo far is sufficient to define an EFCE, as recommendations at infosets $\textsc{y}$, $\textsc{k}$, and $\textsc{l}$ are \emph{not} relevant given that they do not influence players' utilities at the equilibrium ($b$ is never recommended).
However, they become relevant for EFPCEs, since, in perturbed extended games, these infosets could be reached due to a tremble with probability $\eta_t(b)$.
Then, player $2$ must be told to play $p$ at $\textsc{y}$, because her utility is always $1$ if she plays $p$, while it is always $0$ for $o$.
Moreover, with an analogous reasoning, player $1$ must be recommended to play $e$ and $h$ at $\textsc{k}$ and $\textsc{l}$, respectively.
In conclusion, we can state that $\mu \in \Delta_{\Pi}:\mu(aceh,mp) = \mu(aceh,np) = \mu(adeh,mp) = \frac{1}{3}$ is an EFPCE.
%

\paragraph{Properties of EFPCEs}
%
We characterize the relation between EFPCEs and other equilibria, also showing that EFPCEs always exist and represent a refinement of EFCEs.~\footnote{In the following, we denote the sets of equilibria with their corresponding acronyms (\emph{e.g.}, \textsf{NE} is the set of all NEs of a game).}

\begin{restatable}{theorem}{thmrelations}\label{thm:relations}
	This relation holds:
	$
		\textnormal{\textsf{PE}} \subseteq \textnormal{\textsf{EFPCE}} \subseteq \textnormal{\textsf{EFCE}}.
	$
\end{restatable}

\begin{restatable}{theorem}{thmrelationsne}\label{thm:relations_ne}
	The following relations hold:
	\begin{itemize}
		\item $\textnormal{\textsf{EFPCE}} \not\subset \textnormal{\textsf{NE}}$ and $\textnormal{\textsf{NE}} \not\subset \textnormal{\textsf{EFPCE}}$;
		\item $\textnormal{\textsf{EFPCE}} \cap \textnormal{\textsf{NE}} = \textnormal{\textsf{PE}}$.
	\end{itemize}
\end{restatable}

\section{NEs of Perturbed Extended Games}\label{sec:charac_recc_ne}

We provide a characterization of NEs of perturbed extended games $(\Gamma^\textnormal{ext}(\mu), \eta)$, useful for our main algorithmic result on EFPCEs given in the following section.
Specifically, we give a set of easily interpretable conditions which ensure that following recommendations is an NE of $(\Gamma^\textnormal{ext}(\mu), \eta)$.
These are crucial for the derivation of the LP exploited by our algorithm.
Our characterization is inspired by that of EFCEs based on trigger agents (see Lemma~\ref{lem:efce_trigger} in Appendix~\ref{app:trigger}).
However, the presence of trembles in extended games requires some key changes, which we highlight in the following.

First, we introduce some additional notation.
Given a perturbed extended game $(\Gamma^{\textnormal{ext}}(\mu), \eta) $, we let $\xi^\eta(z, \pi)$ be the probability of reaching a node $z \in Z$ when a strategy profile $\pi \in \Pi$ is recommended and players' obey to recommendations, in presence of trembles defined by $\eta$.
Each $\xi^\eta(z, \pi)$ is obtained by multiplying probabilities of actions in $ \sigma(z)$, which are those on the path from the root to $z$.
For each $a \in \sigma(z)$, two cases are possible: either $a$ is prescribed by the recommended $\pi$ and played with its maximum probability given $\eta$, or it is \emph{not}, which means that a tremble occurred with probability $\eta(a)$.
Formally, letting $\mymathbb{1} \{ a \in \pi \}$ be an indicator for the event $a \in \pi$, for every $z \in Z$ and $\pi \in \Pi$:
\begin{equation*}\label{eq:def_xi}
\xi^\eta(z, \pi ) \coloneqq \prod_{a \in A : a \in \sigma(z)}  \eta(a)^{ \mymathbb{1} \{ a \in \pi \} } \tilde{\eta}(a)^{  1 - \mymathbb{1} \{ a \in \pi \}  } p_c(z) ,
\end{equation*}
where, for $a \in A(I)$, we let $\tilde{\eta}(a) \coloneqq 1 - \sum_{a' \neq a \in A(I)} \eta(a')$
be the maximum probability assignable to $a$ given $\eta$.
%
Moreover, for every player $i \in N$, infoset $I \in \I_i$, terminal node $z \in Z(I)$ reachable from $I$, and strategy profile $\pi \in \Pi$, we let $\xi^\eta(z, I, \pi ) $ be defined as $\xi^\eta(z, \pi )$ excluding player $i$'s actions leading from $I$ to $z$, \emph{i.e.}, with the product restricted to actions $a \in \sigma_i(I) \cup  \left(  \sigma(z) \setminus A_i \right)$.
Analogously, for every player $i$'s strategy $\pi_i \in \Pi_i(I)$, we let $\xi^\eta(z, \pi_i )$ be defined for player $i$'s actions $a \in A_i \cap  \left( \sigma(z) \setminus \sigma_i(I) \right)$ from $I$ to $z$.
%

Following recommendations is an NE of the perturbed extended game $(\Gamma^{\textnormal{ext}}(\mu), \eta) $ if, for every player $ i \in N$, infoset $I \in \I_i$, and action $a \in (I)$, player $i$'s utility when obeying to the recommendation $a$ at $I$ is at least as large as the utility achieved by any $(I,a,\hat\mu_i)$-trigger agent.
The fundamental differences with respect to EFCE are: \emph{(i)} an infoset $I$ could be reached even when actions recommended at preceding infosets do not allow it (due to trembles); and \emph{(ii)} trigger agents are subjects to trembles, which means that they may make mistakes while playing the strategy sampled from $\hat \mu_i$.
%

For any terminal node $z \in Z$, the probability of reaching it when following recommendations is:
\begin{equation*}\label{eq:q_pert}
q^\eta_\mu(z) \coloneqq  \sum_{\pi \in \Pi} \xi^\eta(z,\pi)  \mu(\pi) ,
\end{equation*}
where the summation accounts for the probability of reaching $z$ for every possible $\pi$.
The sum is over $\Pi$ rather than $\Pi(z)$ as for EFCE (see Equation~\eqref{eq:q} in Appendix~\ref{app:trigger}), since, due to trembles, $z$ could be reached even when $\pi \notin \Pi(z)$.

For any $(I,a,\hat\mu_i)$-trigger agent, the probability of reaching $z \in Z(I)$ when the agent `gets triggered' is defined as:
\begin{equation*}\label{p_pert}
p_{\mu,\hat\mu_i}^{\eta, I,a}(z) \hspace{-1mm}\coloneqq\hspace{-1.5mm} \left( \hspace{-0.5mm}\sum_{\substack{\pi_i \in \Pi_i(a)\\\pi_{-i} \in \Pi_{-i}}} \hspace{-3.5mm} \xi^\eta(z,I,\pi) \mu (\pi) \hspace{-1.5mm} \right) \hspace{-2mm} \left( \hspace{-0.5mm}\sum_{\hat \pi_i \in \Pi_i(I)} \hspace{-3.5mm} \xi^\eta(z, \hat \pi_i) \hat \mu_i(\hat \pi_i) \hspace{-1.5mm}  \right) \hspace{-1mm},
\end{equation*}
where the first summation is over $\Pi_i(a)$ instead of $\Pi_i(I,a)$ (as in the EFCE, see Equation~\eqref{eq:p} in Appendix~\ref{app:trigger}) since it might be the case that the agent is activated also when the recommended strategy $\pi_i$ does \emph{not} allow to reach infoset $I$.
%
Finally, the overall probability of reaching $z \in Z(I)$ is:
\begin{equation*}\label{y_pert}
y_{\mu,\hat\mu_i}^{\eta, I,a}(z) \coloneqq p_{\mu,\hat\mu_i}^{\eta, I,a}(z)  +  \sum_{\substack{ \pi_i \in \Pi_i \setminus \Pi_i(a) \\ \pi_{-i} \in \Pi_{-i} }} \xi^\eta(z,\pi)  \mu(\pi) ,
\end{equation*}
where the first term is for when the agent `gets triggered', while the second term accounts for the case in which the agent is \emph{not} activated (the two events are independent).

\begin{restatable}{theorem}{thmcharacreccne}\label{thm:charac_recc_ne}
	Given a perturbed extended game $(\Gamma^{\textnormal{ext}}(\mu), \eta) $, following recommendations is an NE of the game if for every $i \in N$ and $(I,a,\hat\mu_i)$-trigger agent for player $i$, it holds that:
	\begin{equation*}\label{eq:charac_recc_ne}
	\sum_{z \in Z(I)} \hspace{-1mm}\left[ \hspace{-1mm}  \left( \sum_{\substack{\pi_i \in \Pi_i(a) \\ \pi_{-i} \in \Pi_{-i} }}\hspace{-2mm} \xi^\eta(z,\pi)  \mu(\pi) \hspace{-1mm}\right) \hspace{-1mm}u_i(z) \hspace{-0.5mm} \right] \hspace{-1.5mm}\geq\hspace{-1.5mm} \sum_{z \in Z(I)} \hspace{-1.5mm}p_{\mu,\hat\mu_i}^{\eta, I,a}(z) u_i(z).
	\end{equation*}
\end{restatable}

\section{Computing an EFPCE in $n$-player EFGs}\label{sec:eah}

We provide a polynomial-time algorithm to compute \emph{an} EFPCE in $n$-player EFGs (also with chance).
%
The algorithm is built on three fundamental components: \emph{(i)} a trembling LP (with exponentially many variables and polynomially many constraints) whose limit solutions define EFPCEs; \emph{(ii)} an adaption of the algorithm by~\citet{farina2018practical} that finds such limit solutions by solving a sequence of (non-trembling) LPs; and \emph{(iii)} a polynomial-time EAH procedure that solves these LPs.

\begin{figure*}[!htp]\centering
	{
		\begin{minipage}[b]{2.7cm}\centering%
			\def\done{.65*1.2}
			\def\dtwo{.50*1.2}
			\def\dleaf{.25*1.2}
			\def\dvert{.65*1.2}
			\begin{tikzpicture}[baseline=-1.1cm,scale=.95]
			\node[fill=black,draw=black,circle,inner sep=.5mm] (A) at (0, 0) {};
			\node[fill=white,draw=black,circle,inner sep=.5mm] (X) at ($(A) + (-\done,-\dvert)$) {};
			\node[fill=white,draw=black,inner sep=.6mm] (l5) at ($(A) + (\done,-\dvert)$) {};
			\node[fill=black,draw=black,circle,inner sep=.5mm] (C) at ($(X) + (-\dtwo, -\dvert)$) {};
			\node[fill=black,draw=black,circle,inner sep=.5mm] (D) at ($(X) + (\dtwo, -\dvert)$) {};
			\node[fill=white,draw=black,inner sep=.6mm] (l1) at ($(C) + (-\dleaf, -\dvert)$) {};
			\node[fill=white,draw=black,inner sep=.6mm] (l2) at ($(C) + (\dleaf, -\dvert)$) {};
			\node[fill=white,draw=black,inner sep=.6mm] (l3) at ($(D) + (-\dleaf, -\dvert)$) {};
			\node[fill=white,draw=black,inner sep=.6mm] (l4) at ($(D) + (\dleaf, -\dvert)$) {};
			\node[inner sep=0] at ($(l1) + (.23,0)$) {\scriptsize$z_1$};
			\node[inner sep=0] at ($(l2) + (.23,0)$) {\scriptsize$z_2$};
			\node[inner sep=0] at ($(l3) + (.23,0)$) {\scriptsize$z_3$};
			\node[inner sep=0] at ($(l4) + (.23,0)$) {\scriptsize$z_4$};
			\node[inner sep=0] at ($(l5) + (.23,0)$) {\scriptsize$z_5$};

			\draw[semithick] (A) --node[fill=white,inner sep=.9] {\scriptsize$a$} (X);
			\draw[semithick] (A) --node[fill=white,inner sep=.9] {\scriptsize$b$} (l5);
			\draw[semithick] (C) --node[fill=white,inner sep=.9] {\scriptsize$c$} (l1);
			\draw[semithick] (C) --node[fill=white,inner sep=.9] {\scriptsize$d$} (l2);
			\draw[semithick] (D) --node[fill=white,inner sep=.9] {\scriptsize$e$} (l3);
			\draw[semithick] (D) --node[fill=white,inner sep=.9] {\scriptsize$f$} (l4);
			\draw[semithick] (X) --(C);
			\draw[semithick] (X) -- (D);
			
			\draw[black!60!white] (A) circle (.2);
			\node[black!60!white]  at ($(A) + (-.4, 0)$) {\textsc{i}};
			
			\draw[black!60!white] (X) circle (.2);
			
			\draw[black!60!white] (C) circle (.2);
			\node[black!60!white]  at ($(C) + (-.2,.38)$) {\textsc{j}};
			
			\draw[black!60!white] (D) circle (.2);
			\node[black!60!white]  at ($(D) + (.2,.38)$) {\textsc{k}};
			
			\end{tikzpicture}
	\end{minipage}}
	\hspace{.5cm}
	{\begin{minipage}[b]{13cm}\small\centering
			\setlength{\tabcolsep}{1pt}
			\begin{tabular}{c@{\hskip 6pt}l|c@{\hskip 6pt}l}
				& \multicolumn{3}{l}{Constraints~\eqref{eq:dev_mu_pert}~and~\eqref{eq:inc_mu_pert} for player $1$, infoset \textsc{i}, and action $a$} \\
				\midrule\\[-4mm]
				& \multicolumn{3}{l}{$u[1,\textsc{i},a] = v[1,\textsc{i},a,\textsc{i}] - \eta(a) w[1,\textsc{i},a,\textsc{i},a]  - \eta(b) w[1,\textsc{i},a,\textsc{i},b] $} \\
				\midrule\\[-4mm]
				$(\textsc{i},a)$: & \multicolumn{3}{l}{$v[1,\textsc{i},a,\textsc{i}] - w[1,\textsc{i},a,\textsc{i},a] \geq v[1,\textsc{i},a,\textsc{j}] + v[1,\textsc{i},a,\textsc{k}] + $} \\
				& \multicolumn{3}{l}{$\quad - \eta(c) w[1,\textsc{i},a,\textsc{j},c]- \eta(d) w[1,\textsc{i},a,\textsc{j},d] - \eta(e) w[1,I\textsc{i},a,\textsc{k},e]- \eta(f) w[1,\textsc{i},a,\textsc{k},f]$} \\
				$(\textsc{i},b)$ & \multicolumn{3}{l}{$v[1,\textsc{i},a,\textsc{i}] - w[1,\textsc{i},a,\textsc{i},b] \geq U_1(z_5) $} \\
				\midrule\\[-4mm]
				$(\textsc{j},c)$: & $v[1,\textsc{i},a,\textsc{j}] - w[1,\textsc{i},a,\textsc{j},c] \geq U_1(z_1)$ \hspace{3mm} & \hspace{3mm}$(\textsc{k},e)$: & $v[1,\textsc{i},a,\textsc{k}] - w[1,\textsc{i},a,\textsc{k},e] \geq U_1(z_3)$ \\
				$(\textsc{j},d)$: & $v[1,\textsc{i},a,\textsc{j}] - w[1,\textsc{i},a,\textsc{j},d] \geq U_1(z_2)$ \hspace{3mm} & \hspace{3mm}$(\textsc{k},f)$: & $v[1,\textsc{i},a,\textsc{k}] - w[1,\textsc{i},a,\textsc{k},f] \geq U_1(z_4)$ \\
				\bottomrule
			\end{tabular}
	\end{minipage}}
	\caption{(\emph{Left}) Simple EFG. (\emph{Right}) Example of Constraints~\eqref{eq:dev_mu_pert}--\eqref{eq:inc_mu_pert}; we let $ U_1(z) \coloneqq  \sum_{ \pi_1 \in \Pi(a)  \pi_2 \in \Pi_2(z) }  \xi^\eta(z,\textsc{i},\pi) \mu[\pi] u_1(z) $ for every $z \in Z$. Variables $v[1,\textsc{i},a,\cdot]$ encode the optimal utility of trigger agents associated to $\textsc{i},a$ at infosets following \textsc{i} (without trembles). Variables $w[1,\textsc{i},a,\cdot,\cdot]$ account for penalties due to trembles. To see this, fix $\mu[\pi]$. Assume that $U_1(z_4) > U_1(z_3)$ and consider the constraints for $(\textsc{k},e)$ and $(\textsc{k},f)$. Then, it must be $v[1,\textsc{i},a,\textsc{k}] = U_1(z_4)$ and $w[1,\textsc{i},a,\textsc{k},f] = 0$, which implies $w[1,\textsc{i},a,\textsc{k},e] = U_1(z_4) - U_1(z_3)$ (constraint for $(\textsc{i},a)$). Similarly, assuming $U_1(z_2) > U_1(z_1)$, it must be $v[1,\textsc{i},a,\textsc{j}] = U_1(z_2)$, $w[1,\textsc{i},a,\textsc{j},d] = 0$, and $w[1,\textsc{i},a,\textsc{j},c] = U_1(z_2) - U_1(z_1)$. An analogous reasoning holds at infosets upwards in the game tree.}\label{fig:example_lp_pert}
\end{figure*}

\paragraph{Trembling LP for EFPCEs}
It resembles the EFCE LP in Problem~\ref{prob:primal_ellip}.
In this case, the constraints appearing in the LP ensure that following recommendations is an NE in a given sequence of perturbed extended games, by exploiting the characterization given in Theorem~\ref{thm:charac_recc_ne}. 
Then, Lemma~\ref{lem:efpce_limit_point} allows to conclude that the limit solutions of the trembling LP define EFPCEs.
In the following, we assume that a sequence of perturbed extended games $\{ (\Gamma^{\textnormal{ext}}(\mu), \eta_t) \}_{t \in \mathbb{N}}$ is given.
For every player $i \in N$, infoset $I \in \I_p$, and action $a \in A(I)$, we introduce a variable $u[i,I,a]$ to encode player $i$'s expected utility when following the recommendation to play $a$ at $I$ in the perturbed extended game $(\Gamma^{\textnormal{ext}}(\mu), \eta_t)$.
These variables are defined by the following constraints:
\begin{align}\label{eq:u_rec_pert}
	& u[i, I, a] = \sum_{z \in Z(I)} \left(  \sum_{\substack{ \pi_i \in \Pi_i(a) \\ \pi_{-i} \in \Pi_{-i}} } \xi^\eta(z, \pi) \mu[\pi] \right) u_i(z) \\
	& \hspace{4.1cm} \forall  i \in N, \forall I \in \I_i, \forall a \in A(I). \nonumber
\end{align}
Then, we introduce constraints that recursively define variables $v[i,I,a,J]$ for every infoset $J \in \I_i : I \preceq J$. These encode the maximum expected utility obtained at infoset $J$ by trigger agents associated with $I$ and $a$.
%
To this end, we also need some auxiliary non-negative variables $w[i,I,a,J,a']$, which are defined for every player $i \in N$, infoset $I \in \I_i$, action $a \in A(I)$, infoset $J \in \I_i : I \preceq J$ following $I$ (this included), and action $a' \in A(J)$ available at $J$.
%
\begin{align}
	& v[i,I,a,J] - w[i,I,a,J,a'] \geq \label{eq:dev_mu_pert} \\
	& \hspace{0.3cm} \sum_{z \in Z^\bot(J,a^\prime)  } \left(   \sum_{\substack{ \pi_i \in \Pi_i(a) \\ \pi_{-i} \in \Pi_{-i}(z)  } } \xi^\eta(z,I, \pi) \mu[\pi]  \right) u_i(z)  + \nonumber \\
	& \hspace{0.3cm}\sum_{K \in \C(J, a^\prime) } \hspace{-1.5mm} \left(  v[i,I,a,K] -\hspace{-1.5mm}  \sum_{a'' \in A(K)} \hspace{-1.5mm}  \eta_t(a'') w[i,I,a,K,a'']  \right) \nonumber\\
	& \forall i \in N, \forall I \in \I_i, \forall a \in A(I), \forall J \in \I_i: I \preceq J,  \forall a' \in A(J). \nonumber
\end{align}
Intuitively, each auxiliary variable $w[i,I,a,J,a']$ represents a penalty on $v[i,I,a,J]$ due to the possibility of trembling by playing a (possibly) sub-optimal action $a' \in A(J)$ at $J$.
Indeed, whenever $a'$ is an optimal action at infoset $J$, then $w[i,I,a,J,a']$ is set to $0$ in any solution; otherwise, $w[i,I,a,J,a']$ represents how much utility is lost by playing $a'$ instead of an optimal action (see Figure~\ref{fig:example_lp_pert}(\emph{Right}) for an example).
Finally, the incentive constraints are:
\begin{align}\label{eq:inc_mu_pert}
	& u[i,I,a] = v[i,I,a,I] - \sum_{a' \in A(I)} \eta_t(a') w[i,I,a,I,a'] \\
	&  \hspace{4cm} \forall i \in N, \forall I \in \I_i, \forall a \in A(I). \nonumber
\end{align}
Figure~\ref{fig:example_lp_pert} provides an example of Constraints~\eqref{eq:dev_mu_pert}~and~\eqref{eq:inc_mu_pert} to better clarify their meaning.
The following theorem shows that Constraints~\eqref{eq:u_rec_pert},~\eqref{eq:inc_mu_pert},~and~\eqref{eq:dev_mu_pert} correctly encode the conditions given in Theorem~\ref{thm:charac_recc_ne}, which ensure that following recommendations is an NE in $(\Gamma^{\textnormal{ext}}(\mu), \eta_t)$.
\begin{restatable}{theorem}{thmtremblinglptheorem}\label{thm:trembling_lp_theorem}
	Given a perturbed extended game $(\Gamma^{\textnormal{ext}}(\mu), \eta_t)$, if Constraints~\eqref{eq:u_rec_pert},~\eqref{eq:dev_mu_pert},~and~\eqref{eq:inc_mu_pert} can be satisfied for the vector $\boldsymbol{\mu}$ of variables $\mu[\pi]$ encoding the distribution $\mu$, then following recommendations is an NE of $(\Gamma^{\textnormal{ext}}(\mu), \eta_t)$.
\end{restatable}
\noindent
By substituting the expression of $u[i,I,a]$ (given by Constraints~\eqref{eq:u_rec_pert}~and~\eqref{eq:inc_mu_pert}) into Constraints~\eqref{eq:dev_mu_pert}, we can formulate the following trembling LP parameterized by $t \in \mathbb{N}$:
\begin{subequations}\label{prob:primal_ellip_pert}
	\begin{align}
	\max_{\boldsymbol{\mu} \geq \boldsymbol{0}, \boldsymbol{v}, \boldsymbol{w} \geq \boldsymbol{0}} & \quad \sum_{\pi \in \Pi} \mu[\pi] \quad \textnormal{s.t.} \\
	& A_t \boldsymbol{\mu} + B  \boldsymbol{v} + C_t \boldsymbol{w} \geq \boldsymbol{0}, \label{eq:primal_ellip_pert_cons}
	\end{align}
\end{subequations}
where $A_t $ is the analogous of matrix $A$ in Problem~\ref{prob:primal_ellip}, $\boldsymbol{w}$ is a vector whose components are the variables $w[i,I,a,J,a']$, and $C_t$ is a matrix defining the constraints coefficients for these variables.
Notice that the coefficients of variables in $\boldsymbol{v }$ (as defined by $B$) are the same as in Problem~\ref{prob:primal_ellip}.

\paragraph{Limit Solutions of Trembling LP}
%
%
%
Problem~\ref{prob:primal_ellip_pert} can be cast into the framework of~\citet{farina2018practical} by defining sequences of lower bounds $\eta_t$ by means of vanishing polynomials in a parameter $\epsilon \rightarrow 0$.
As a result, the polynomial-time algorithm by~\citet{farina2018practical} can be used, with the only difference that, at each step, for a fixed value of the parameter $\epsilon$ (\emph{i.e.}, particular lower bounds $\eta_t$), it needs to solve an instance of Problem~\ref{prob:primal_ellip_pert} featuring exponentially many variables.
Provided that the latter can be done in polynomial time, the polynomiality of the overall procedure is preserved, since the bounds on the running time provided by~\citet{farina2018practical} do not depend on the number of variables in the LP. 

\paragraph{EAH Procedure}
In order to solve Problem~\ref{prob:primal_ellip_pert} for a particular lower bound function $\eta_t$ in polynomial time, we can apply a procedure similar to the EAH algorithm by~\citet{huang2008computing}.
Notice that Problem~\ref{prob:primal_ellip_pert} is always unbounded, since there always exists a distribution $\mu \in \Delta_\Pi$ such that following recommendations is an NE of the perturbed extended game $(\Gamma^{\textnormal{ext}}(\mu), \eta_t)$ (such $\mu$ is an EFCE of the corresponding perturbed, non-extended game).
Thus, we only need to provide a polynomial-time separation oracle for the always-infeasible dual of Problem~\ref{prob:primal_ellip_pert}, which reads as:
\begin{subequations}\label{prob:dual_ellip_pert}
	\begin{align}
	A_t^\top \boldsymbol{y} \hspace{0.3cm} &\leq \hspace{0.3cm}-\boldsymbol{1} \\
 	B^\top \boldsymbol{y} \hspace{0.3cm}&= \hspace{0.3cm}\textcolor{white}{-}\boldsymbol{0} \\
	C_t^\top \boldsymbol{y} \hspace{0.3cm}&\geq \hspace{0.3cm}\textcolor{white}{-}\boldsymbol{0} \\
	\boldsymbol{y} \hspace{0.3cm}&\geq \hspace{0.3cm}\textcolor{white}{-} \boldsymbol{0},
	\end{align}
\end{subequations}
where the vector of dual variables $\boldsymbol{y}$ has the same role as in Problem~\ref{prob:dual_ellip}, since the constraints of the primal problems are indexed on the same sets.
Notice that constraints $C_t^\top \boldsymbol{y} \geq \boldsymbol{0}$ are polynomially many.
As a result, one can always check whether one of these constraints is violated in polynomial time and, if this is the case, output one such constraint as a violated inequality.
This allows to focus on separation oracles for the other constraints. 
Then, the required one is given by the following lemma, an analogous of Lemma~\ref{lem:sep_ro}.
\begin{restatable}{lemma}{lemmasep}\label{lem:lemmasep}
	If $\boldsymbol{y} \geq \boldsymbol{0}$ is such that $B^\top \boldsymbol{y} = \boldsymbol{0}$, then there exists $\boldsymbol{\mu}$ encoding a product distribution $\mu \in \Delta_{\Pi}$ such that $\boldsymbol{\mu}^\top A_t^\top \boldsymbol{y} = 0$.
	Moreover, $\boldsymbol{\mu}$ can be computed in poly-time.
\end{restatable}
\noindent
The proof of Lemma~\ref{lem:lemmasep} follows the same line as that of Lemma~5 by~\citet{huang2008computing} (see~\citep{huang2011equilibrium} for its complete version) and it is based on the CE existence proof by~\citet{hart1989existence}.

\section{Discussion and Future Works}\label{sec:discussion}

We started the study of \emph{trembling-hand perfection} in sequential games with correlation, introducing the EFPCE as a refinement of the EFCE that amends its weaknesses off the equilibrium path.
This paves the way to a new research line, raising novel game-theoretic and computational challenges.

As for EFPCEs, an open question is whether compact correlated strategy representations, like the EFCE-based \emph{correlation plan} by~\citet{von2008extensive}, are possible in some restricted settings, such as $2$-player games without chance.
This would enable the optimization over the set of EFPCEs in polynomial time.
The main challenge raised by EFPCEs with respect to EFCEs is that the former require to reason about general, un-reduced strategy profiles.

Another possible future work is to extend our analysis to other CE-based solution concepts, such as the \emph{normal-form} CE and the \emph{agent-form} CE (see~\citep{von2008extensive} for their definitions).
This raises the interesting question of how different trembling-hand-based CEs are able to amend weaknesses off the equilibrium path.

Finally, an interesting direction is to consider different ways of refining CE-based equilibria in sequential games, such as, \emph{e.g.}, using \emph{quasi-perfection}~\citep{van1984relation}.

\section*{Acknowledgments}	
This work has been partially supported by the Italian MIUR PRIN 2017 Project ALGADIMAR ``Algorithms, Games, and Digital Market''.

\bibliography{refs}
\bibliographystyle{aaai}

\clearpage

\setcounter{secnumdepth}{1} 
\appendix
\onecolumn

\begin{center}
	\LARGE\bf Appendix
\end{center}

\section{Sequence-Form Representation for Perfect-Recall EFGs}\label{app:sequence}

The number of pure strategies $|\Pi_i|$ of each player $i \in N$ may be exponentially large in the size of an EFG, preventing the development of scalable computational tools using them.
Moreover, the same holds for \emph{reduced} pure strategies, which only specify actions at infosets that are reachable given the player's past moves. 
This problem is circumvented by the \emph{sequence form} introduced by~\citet{von1996efficient}, where each player selects a \emph{sequence} of actions rather than a pure strategy.
For any node $h \in H$, we let $\sigma_i(h)$ be the ordered sequence of actions of player $i \in N$ on the path from the root of the game tree to $h$.
We recall that, given the perfect recall assumption, all nodes in an infoset $I \in \I_i$ of player $i \in N$ define the same sequence $\sigma_i(I)$ of player $i$'s actions, \emph{i.e.}, it holds $\sigma_i(h) = \sigma_i(I) $ for all $h \in I$.
Moreover, $\sigma_i(I)$ can be {extended} by any action $a \in A(I)$, defining a new player $i$'s sequence $\sigma_i(I) a$.
Thus, by introducing the empty sequence to represent the paths in the game tree in which a player does not play, the set of sequences available to player $i \in N$ is $\Sigma_i \coloneqq \{\varnothing\} \cup \{ \sigma_i(I)a \mid I \in \I_i , a \in A(I) \}$.
Within the sequence form, mixed strategies are expressed as \emph{realization plans}.
A realization plan for player $i \in N$ is a function $x_i: \Sigma_i \to [0,1]$, with $x_i(\sigma_i)$ expressing the realization probability of sequence $\sigma_i \in \Sigma_i$.
In order to be well defined, $x_i$ must satisfy the linear constraints $x_i(\varnothing) = 1$ and $x_i(\sigma_i(I)) = \sum_{a \in A(I)} x_i(\sigma_i(I)a)$ for every infoset $I \in \I_i$.
Since the number of sequences $|\Sigma_i|$ of each player $i \in N$ is polynomial in the size of an EFG and realization plans can be easily expressed by linear constraints, the sequence form is an appealing formalism for handling EFGs.
Moreover, as shown by~\citet{von1996efficient}, the crucial property of the sequence form is that realization plans and behavior strategies are equally expressive in EFGs with perfect recall.
In particular, $x_i$ is equivalent to a behavior strategy that selects $a \in A(I)$ with probability $\frac{x_i(\sigma_i(I) a)}{x_i(\sigma_i(I))}$ if $x_p(\sigma_i(I)) > 0$ and arbitrarily if $x_i(\sigma_i(I)) = 0$.
Conversely, a behavior strategy $\beta_i$ is equivalent to a realization plan that selects each sequence $\sigma_i \in \Sigma_i$ with probability $\prod_{a \in \sigma_i} \beta_i(a)$.

\section{Characterization of EFCEs Using Trigger Agents}\label{app:trigger}

We provide a formal statement of the characterization of EFCEs based on trigger agents (see Definition~\ref{def:trigger}), originally introduced by~\citet{DBLP:conf/icml/GordonGM08}~and~\citet{DBLP:conf/nips/FarinaLFS19a} (see also~\citep{DBLP:conf/aaai/FarinaBS20} for a more general treatment).
We recall that such characterization is based on the fact that $\mu \in \Delta_\Pi$ is an EFCE if, for every $i \in N$, player $i$'s expected utility when following recommendations is at least as large as the expected utility that any $(I, a, \hat \mu_i)$-trigger agent for player $i$ can achieve (assuming the opponents' do not deviate from recommendations).

For any $\mu\in\Delta_\Pi$ and $(I, a, \hat \mu_i)$-trigger agent, we define the probability of reaching  a terminal node $z \in Z(I)$ as:
\begin{equation}\label{eq:p}
p_{\mu,\hat\mu_i}^{I,a}(z) \coloneqq \left( \sum_{\substack{\pi_i \in \Pi_i(I, a)\\\pi_{-i} \in \Pi_{-i}(z)}} \mu (\pi_i, \pi_{-i}) \right)  \left( \sum_{\hat \pi_i \in \Pi_i(z)} \hat \mu_i(\hat \pi_i)  \right) p_c(z),
\end{equation}
which accounts for the fact that the agent follows recommendations until she receives the recommendation of playing $a$ at $I$, and, thus, she `gets triggered' and plays according to $\hat \pi_i$ sampled from $\hat \mu_i$ from $I$ onwards.
Moreover, the probability of reaching a terminal node  $z \in Z$ when following the recommendations is defined as follows:
\begin{equation}\label{eq:q}
q_\mu(z) \coloneqq \left( \sum_{\pi \in \Pi(z)} \mu(\pi) \right) p_c(z).
\end{equation}

Then, the following lemma provides the trigger-agent-based characterization of EFCEs:
\begin{lemma}[\citet{DBLP:conf/aaai/FarinaBS20}]\label{lem:efce_trigger}
	Given an EFG $\Gamma$, $\mu \in \Delta_\Pi$ is an EFCE of $\Gamma$ if for every $i \in  N$ and $(I, a, \hat \mu_i)$-trigger agent for player $i$, it holds that:
	\begin{equation*}\label{eq:efce_simp}
	\sum_{z \in Z(I,a)} q_\mu(z) u_i(z) \geq \sum_{z \in Z(I)} p_{\mu,\hat\mu_i}^{I,a}(z) u_i(z).
	\end{equation*}
\end{lemma}

\section{LP Formulation for the Set of EFCEs in $n$-Player EFGs}\label{app:lp}

We show how to derive the LP formulation (Problem~\ref{prob:primal_ellip}) for the set of EFCEs in $n$-player EFGs originally introduced by~\citet{huang2008computing}, using the characterization of EFCEs based on trigger agents (see Definition~\ref{def:trigger} and Lemma~\ref{lem:efce_trigger}).

In the following, we assume that a probability distribution $\mu \in \Delta_{\Pi}$ is encoded by means of variables $\mu[\pi]$, defined for $\pi \in \Pi$.
For every player $i \in N$, infoset $I \in \I_i$, and action $a \in A(I)$, we introduce a variable $u[i, I, a]$ representing player $i$'s expected utility when following the recommendation to play $a$ at infoset $I$.
These variables are defined by the following constraints:
\begin{align}\label{eq:aux_u}
& u[i, I, a] = \sum_{z \in Z(I,a)} \left(  \sum_{\pi \in \Pi(z)} \mu[\pi] \right) p_c(z) u_i(z) & \forall  i \in N, \forall I \in \I_i, \forall a \in A(I).
\end{align}
Then, we need to introduce constraints which ensure that following recommendations guarantees a utility at least as large as that achieved by any $(I, a, \hat \mu_i)$-trigger agent.
For every infoset $J \in \I_i$ such that $I \preceq J$, we introduce a variable $v[i,I,a,J]$ that encodes the maximum expected utility obtained at infoset $J$ by trigger agents associated with $I$ and $a$.
We can recursively define variables $v[i,I,a,J]$ as follows:
\begin{align}\label{eq:dev_mu}
& v[i,I,a,J] \geq \sum_{z \in Z^\bot (J, a')} \left(   \sum_{\substack{ \pi_i \in \Pi_i(I, a) \\ \pi_{-i} \in \Pi_{-i}(z)  } } \mu[\pi_i, \pi_{-i}]  \right) p_c(z) u_i(z) + \sum_{K \in \C(J, a^\prime)} v[i, I, a, K] \\
& \hspace{7.5cm}\forall i \in N, \forall I \in \I_i, \forall a \in A(I), \forall J \in \I_i: I \preceq J,  \forall a^\prime \in A(J), \nonumber
\end{align}
where we notice that the first summation is over the set of terminal nodes which are reachable from $J$ by playing $a^\prime$ without traversing any other player $i$'s infoset.
The following incentive constraints complete the formulation:
\begin{align}\label{eq:inc_mu}
& u[i,I,a] = v[i,I,a,I] & \forall i \in N, \forall I \in \I_i, \forall a \in A(I). 
\end{align}

A direct application of Lemma~\ref{lem:efce_trigger} and LP duality is enough to prove that Constraints~\eqref{eq:aux_u},~\eqref{eq:dev_mu},~and~\eqref{eq:inc_mu} correctly characterize the set of EFCEs (formally, it is enough to follow steps similar to those in the proof of Theorem~\ref{thm:trembling_lp_theorem}, with the only difference that Constraints~\eqref{eq:primal_inner_cons3} in the inner maximization problems and the corresponding dual variables $w[i,I,a,J,a']$ are missing).

By substituting the equalities in Constraints~\eqref{eq:aux_u}~and~\eqref{eq:inc_mu} into Constraints~\eqref{eq:dev_mu}, we obtain the following set of linear constraints, which are equivalent to those introduced by~\citet{huang2008computing}:
\begin{equation*}
	A \boldsymbol{\mu} + B \boldsymbol{v }\geq \boldsymbol{0},
\end{equation*}
where $\boldsymbol{\mu}$ is the vector whose components are the variables $\mu[\pi]$ for $\pi \in \Pi$, while $\boldsymbol{v}$ is the vector of variables $v[i,I,a,J]$ indexed by $ i \in N,  I \in \I_i,  a \in A(I)$, and $J \in \I_i: I \preceq J$.
Moreover, the matrices $A$ and $B$ encode the coefficients appearing in Constraints~\eqref{eq:aux_u}~and~\eqref{eq:dev_mu}.
Specifically, non-zero entries of $A$ are products $p_c(z) u_i(z)$, while those of $B$ are either $1$ or $-1$.

\section{Discussion on EFPCEs and Un-Reduced Strategies}\label{app:reduced}

Next, we discuss the reasons why EFPCEs need un-reduced strategy profiles in order to be defined consistently.

First, we remark that, as discussed by~\citet{von2008extensive}, restricting the definition of probability distributions $\mu$ to \emph{reduced} strategy profiles (\emph{i.e.}, those in which each player's pure strategy only specifies actions at infosets reachable given that player's moves, see~\citep{vermeulen1998reduced} for a formal definition) is sufficient for the characterization of the classical notions of correlated equilibria.
Intuitively, the reason is that, at the equilibrium, each player follows recommendations issued by the correlation device, and, thus, the latter does \emph{not} need to specify action recommendations for the player at those infosets that are never reached when following recommendations at the preceding infosets of the same player.

This is no longer the case if we introduce trembles in the game, which make all the infosets reachable with positive probability even when committing to following recommendations.
As a result, the correlation device has to be ready to issue action recommendations everywhere in the game. 
Then, when defining EFPCEs, we cannot restrict the attention to probability distributions over reduced strategy profiles, and un-reduced ones are necessary.
The EFG in Figure~\ref{fig:example_game}(\emph{Left}) provides an example where un-reduced strategies are necessary to express EFPCEs.
As shown in the main text, any EFPCE of the game must recommend to play $a$ at $\textsc{i}$, while, at the same time, it is crucial to define recommendations also at infosets $\textsc{k}$ and $\textsc{l}$, in order to achieve optimality off the equilibrium path.
Clearly, this is incompatible with reduced strategies, as any player $1$'s reduced strategy prescribing $a$ at $\textsc{i}$ does not specify anything at infosets $\textsc{k}$ and $\textsc{l}$, which are unreachable when playing $a$ at $\textsc{i}$.

\section{Detailed Examples of EFPCEs}\label{app:example}

Consider the EFG in Figure~\ref{fig:example_game}(\emph{Left}) and lower bound functions $\eta_t: A \to (0,1)$ for $t \in \mathbb{N}$, with $\eta_t(a)$ converging to zero as $t \rightarrow \infty$ for each $a \in A$.
First, let us notice that, without trembles, player $1$ is always better off playing action $a$ at the root infoset $\textsc{i}$, since she can guarantee herself a utility of $1$ by selecting $c$ at the following infoset $\textsc{j}$, while she can achieve at most a utility of $\frac{1}{2}$ by playing $b$.
Thus, any EFPCE of the game (as well as any EFCE) must recommend $a$ at $\textsc{i}$ with probability $1$, since there is no way player $1$ can be incentivized to play $b$.
Then, in the sub-game reached when playing $a$ at $\textsc{i}$, it is easy to check that recommending the pairs of actions $(c,m)$, $(c,n)$, and $(d,m)$ each with probability $\frac{1}{3}$ is an equilibrium, as each player has no incentive to deviate from each possible recommendation, even in presence of trembles.
As an example, consider the case in which player $1$ is told to play action $c$ at $\textsc{j}$.
Then, by following the recommendations, she gets a utility equal to:
\begin{align*}
	& \Big[  2 \cdot \frac{1}{3}  \left(  1 - \eta_t(n) \right) \left( 1 - \eta_t(d) \right) + 3 \cdot \frac{1}{3}   \left(  1 - \eta_t(n) \right) \eta_t(d) + 1 \cdot \frac{1}{3}  \eta_t(n) \left( 1 - \eta_t(d) \right) + 0 \cdot \frac{1}{3}  \eta_t(n) \eta_t(d)  \Big] + \\
	& \Big[  1 \cdot \frac{1}{3}  \left(  1 - \eta_t(m) \right) \left( 1 - \eta_t(d) \right) + 0 \cdot \frac{1}{3}   \left(  1 - \eta_t(m) \right) \eta_t(d) + 2 \cdot \frac{1}{3}  \eta_t(m) \left( 1 - \eta_t(d) \right) + 3 \cdot \frac{1}{3}  \eta_t(m) \eta_t(d)  \Big],
\end{align*}
where the first sum is for the case in which $(c,m)$ is recommended, while the second one is for $(c,n)$.
Each term appearing in a sum is for one of the four possible outcomes that may result when following recommendations subject to trembles.
Instead, player $1$'s utility if deviating to $d$ at $\textsc{j}$ is:
\begin{align*}
	&  \Big[  3 \cdot \frac{1}{3}  \left(  1 - \eta_t(n) \right) \left( 1 - \eta_t(c) \right) + 2 \cdot \frac{1}{3}   \left(  1 - \eta_t(n) \right) \eta_t(c) + 0 \cdot \frac{1}{3}   \eta_t(n) \left( 1 - \eta_t(c) \right) + 1 \cdot \frac{1}{3}   \eta_t(n) \eta_t(c)  \Big] + \\
	&  \Big[  0 \cdot \frac{1}{3}  \left(  1 - \eta_t(m) \right) \left( 1 - \eta_t(c) \right) + 1 \cdot \frac{1}{3}   \left(  1 - \eta_t(m) \right) \eta_t(c) + 3 \cdot \frac{1}{3}   \eta_t(m) \left( 1 - \eta_t(c) \right) + 2 \cdot \frac{1}{3}   \eta_t(m) \eta_t(c)  \Big].
\end{align*}
A simple calculation shows that the first quantity is greater than or equal to the second one as the lower bounds approach zero.
Analogous conditions hold for other recommendations at infosets $\textsc{x}$ and $\textsc{j}$.
Notice that, when lower bounds are zero, the conditions above collapse to the classical incentive constraints for EFCE.
The correlation device described up to this point is sufficient to define an EFCE, as recommendations at infosets $\textsc{y}$, $\textsc{k}$, and $\textsc{l}$ are \emph{not} relevant given that they do not influence players' utilities at the equilibrium ($b$ is never recommended).
However, in perturbed extended games, these infosets could be reached due to a tremble which happens with probability $\eta_t(b)$, and, thus, recommendations at such infosets become relevant.
Then, it is easy to check that player $2$ must be told to play $p$ at $\textsc{y}$, because her utility is always $1$ if she plays $p$, while it is always $0$ when playing $o$.
Moreover, with an analogous reasoning, player $1$ must be recommended to play $e$ and $h$ at $\textsc{k}$ and $\textsc{l}$, respectively.
As an example, consider the case in which player $1$ is recommended to play $e$ at $\textsc{k}$.
Then, her utility would be $ \frac{1}{2} \cdot \left( 1 - \eta_t(f) \right) + 0 \cdot \eta_t(f)$, while she would get $0 \cdot \left( 1 - \eta_t(e) \right) + \frac{1}{2} \cdot \eta_t(e)$ by deviating to $e$.
Similar conditions hold for infosets $\textsc{y}$ and $\textsc{l}$.
In conclusion, we can state that the following distribution $\mu \in \Delta_{\Pi}$ defines an EFPCE:
\begin{align*}
\mu(aceh,mp) = \mu(aceh,np) = \mu(adeh,mp) = \frac{1}{3}.
\end{align*}
Let us remark that this is not the only EFPCE of the game, as there are other ways of correlating players' behavior at infosets $\textsc{x}$ and $\textsc{j}$ while satisfying the required incentive constraints.
For example, setting
\begin{align*}
	\mu(aceh,mp) = \mu(aceh,np) = \mu(adeh,mp) = \mu(adeh,np) = \frac{1}{4}
\end{align*}
defines a valid EFPCE that results from a PE of the game (where players play uniform strategies at infosets $\textsc{x}$ and $\textsc{j}$).

\section{Proofs of Theorems and Lemmas}\label{app:proofs}

In this section, we provide the complete proofs of Theorems~\ref{thm:relations},~\ref{thm:relations_ne},~\ref{thm:charac_recc_ne},~\ref{thm:trembling_lp_theorem},~and~Lemma~\ref{lem:lemmasep}.

\thmrelations*

\begin{proof}
	Clearly, $\textnormal{\textsf{EFPCE}} \subseteq \textnormal{\textsf{EFCE}}$ holds since any PE of $\Gamma^{\textnormal{ext}}(\mu)$ is also an NE.
	As for the other relation, let $\{ \beta_i \}_{i \in N}$ be a PE of $\Gamma$ obtained for a sequence of perturbed games $\{ (\Gamma, \eta_t) \}_{t \in \mathbb{N}}$ and a corresponding sequence of NEs in these games, namely $\{ \beta_{i,t} \}_{i \in N}$ for $t \in \mathbb{N}$, where each $\beta_{i,t}$ is a well-defined player $i$'s behavior strategy in $(\Gamma, \eta_t) $, \emph{i.e.}, it holds $\beta_{i,t}(a) \geq \eta_t(a)$ for all $t \in \mathbb{N}$, $i \in N$, and $a \in A_i$.
	Let $\mu \in \Delta_{\Pi}$ be such that, for every $\pi \in \Pi$, it holds $\mu(\pi) = \prod_{i \in N} \prod_{I \in \I_i} \beta_i(\pi_i(I))$.
	Consider the extended game $\Gamma^{\textnormal{ext}}(\mu)$, where we denote with $\I^{\textnormal{ext}}_i$ the set of all infosets of player $i \in N$, one for each infoset $I \in \I_i$ of $\Gamma$ and possible combination of recommendations received by $i$ at the infosets $J \in \I_i: J \preceq I$.
	Overloading the notation, for each infoset $I \in \I^{\textnormal{ext}}_i$ of the extended game, we use $I$ as well to denote the corresponding infoset in the original game.
	We also use $A(I)$ as the set of actions available at $I \in \I^{\textnormal{ext}}_i$.
	Let $\{ (\Gamma^{\textnormal{ext}}(\mu), \eta_t) \}_{t \in \mathbb{N}}$ be the sequence of perturbed extended games resulting from $\{ (\Gamma, \eta_t) \}_{t \in \mathbb{N}}$.
	Furthermore, for each $t \in \mathbb{N}$ and player $i \in N$, we define a player $i$'s behavior strategy for $(\Gamma^{\textnormal{ext}}(\mu), \eta_t) $ such that, at each infoset $I \in \I^{\textnormal{ext}}_i$:
	\begin{itemize}
		\item all the residual probability given the lower bounds $1 - \sum_{a \in A(I): a \neq \pi_i(I)} \eta_t(a)$ is placed on the action $\pi_i(I)$ which is recommended at $I$; and
		\item all the other, non-recommended actions $a \in A(I): a \neq \pi_i(I)$ are played with probabilities equal to their corresponding lower bounds $\eta_t(a)$.
	\end{itemize}
	Intuitively,  these strategies encode the fact that players follow recommendations in the perturbed extended games $(\Gamma^{\textnormal{ext}}(\mu), \eta_t) $, where trembles prevent them to perfectly obey to recommendations.
	Given the definition of $\mu$ and the fact that each $\{ \beta_{i,t} \}_{i \in N} $ constitutes an NE for the perturbed game $(\Gamma, \eta_t) $, we can conclude that the behavior strategies defined above constitute NEs for the perturbed extended games $ (\Gamma^{\textnormal{ext}}(\mu), \eta_t) $.
	Thus, any limit point of the sequence defined by such behavior strategies for $t \in \mathbb{N}$ is a PE of $\Gamma^{\textnormal{ext}}(\mu)$.
	Moreover, by definition, any limit point prescribes to play recommended actions, which shows that $\mu$ defines an EFPCE of $\Gamma$, proving that $\textnormal{\textsf{PE}} \subseteq \textnormal{\textsf{EFPCE}}$.
\end{proof}

\thmrelationsne*

\begin{proof}
	Let us start with the first bullet point.
	We consider the EFG in Figure~\ref{fig:example_game}(\emph{Left}) in order to provide examples that prove the two relations.
	Notice that, in such game, player~$1$ is always better off playing action $a$ at the first infoset $\textsc{i}$, since she can guarantee herself to get at least $1$ by playing $c$ at $\textsc{j}$, while she can achieve at most $1$ by playing action $b$.
	Then, it is easy to check that, in any NE of the game, the players play behavior strategies $\beta_1$ and $\beta_2$ such that:
	\begin{itemize}
		\item $\beta_1(a) = 1$ and $\beta_1(b) = 0$, while $\beta_1(c) = \beta_1(d) = \frac{1}{2}$; and
		\item $\beta_2(m) = \beta_2(n) = \frac{1}{2}$.
	\end{itemize}
	The players' behavior at other infosets can be any, as it does not affect players' utilities at the equilibrium (given that infosets $\textsc{y}$, $\textsc{k}$, and $\textsc{l}$ are never reached due to $\beta_1(b)=0$).
	As we have shown in the main text, one EFPCE of the game is the distribution $\mu \in \Delta_\Pi$ such that
	\begin{align*}
	\mu(aceh, mp) = \mu(adeh,mp) = \mu(aceh, np) = \frac{1}{3},
	\end{align*}
	which enforces each player to follow recommendations, even in presence of trembles.
	Clearly, this distribution $\mu$ cannot come up from players' behavior strategies, and, thus, it cannot result from an NE.
	This shows that $\textnormal{\textsf{EFPCE}} \not\subset \textnormal{\textsf{NE}}$.
	Moreover, notice that any NE such that $\beta_1(f) > 0$ cannot determine a distribution $\mu \in \Delta_{\Pi}$ which is an EFPCE, since it would be the case that action $f$ is recommended with positive probability when reaching infoset $\textsc{k}$ (due to trembles).
	However, player $1$ cannot have any incentive to follow such recommendation, as she can gain a utility of $1$ instead of $0$ by deviating to $e$.
	This proves that $\textnormal{\textsf{NE}} \not\subset \textnormal{\textsf{EFPCE}}$.
	
	As for the second bullet point, notice that all the EFPCEs $\mu \in \Delta_{\Pi}$ which are also NEs must be such that $\mu$ is obtained from some players' behavior strategies defining an NE.
	As a result, by definition of EFPCE, we can conclude that such behavior strategies are indeed PEs.
\end{proof}

\thmcharacreccne*

\begin{proof}
	Given the definitions of $q^\eta_\mu(z)$, $p_{\mu,\hat\mu_i}^{\eta, I,a}(z)$, and $y_{\mu,\hat\mu_i}^{\eta, I,a}(z)$, following recommendations is an NE of $(\Gamma^{\textnormal{ext}}(\mu), \eta) $ if for every $i \in N$ and $(I,a,\hat\mu_i)$-trigger agent for player $i$, it holds that:
	\begin{align*}
	\sum_{z \in Z} q^\eta_\mu(z) u_i(z) \geq \sum_{z \in Z \setminus Z(I)} q^\eta_\mu(z) u_i(z) +  \sum_{z \in Z(I)} y_{\mu,\hat\mu_i}^{\eta, I,a}(z) u_i(z).
	\end{align*}
	Equivalently, we can write:
	\begin{align*}
	& \sum_{z \in Z(I)} q^\eta_\mu(z) u_i(z) \geq \sum_{z \in Z(I)} y_{\mu,\hat\mu_i}^{\eta, I,a}(z) u_i(z) \\
	& \sum_{z \in Z(I)} \left[   \sum_{\substack{\pi_i \in \Pi_i(a) \\ \pi_{-i} \in \Pi_{-i} }} \xi^\eta(z,\pi)  \mu(\pi) +  \sum_{\substack{\pi_i \in \Pi_i \setminus \Pi_i(a) \\ \pi_{-i} \in \Pi_{-i} }} \xi^\eta(z,\pi)  \mu(\pi)  \right]  u_i(z) \geq \\
	& \hspace{3cm} \geq  \sum_{z \in Z(I)} p_{\mu,\hat\mu_i}^{\eta, I,a}(z)  u_i(z) + \sum_{z \in Z(I)} \left[  \sum_{\substack{\pi_i \in \Pi_i \setminus \Pi_i(a) \\ \pi_{-i} \in \Pi_{-i} }} \xi^\eta(z,\pi)  \mu(\pi)  \right] u_i(z) \\
	& \sum_{z \in Z(I)} \left[   \left( \sum_{\substack{\pi_i \in \Pi_i(a) \\ \pi_{-i} \in \Pi_{-i} }} \xi^\eta(z,\pi)  \mu(\pi) \right) u_i(z)   \right] \geq \sum_{z \in Z(I)} p_{\mu,\hat\mu_i}^{\eta, I,a}(z) u_i(z),
	\end{align*}
	which proves the result.
\end{proof}

\thmtremblinglptheorem*

\begin{proof}
	By Theorem~\ref{thm:charac_recc_ne}, following recommendations is an NE of $(\Gamma^{\textnormal{ext}}(\mu), \eta_t)$ if the vector $\boldsymbol{\mu}$ of variables $\mu[\pi]$ encoding the distribution $\mu$ satisfies the following constraints (form here on, we omit the subscript $t$ for the ease of notation):
	\begin{align*}
	& \sum_{z \in Z(I)} \left[   \left( \sum_{\substack{\pi_i \in \Pi_i(a) \\ \pi_{-i} \in \Pi_{-i} }} \xi^\eta(z,\pi)  \mu(\pi) \right) u_i(z)   \right] = \sum_{z \in Z(I)} p_{\mu,\hat\mu_i^{I,a}}^{\eta, I,a}(z) u_i(z) &  \forall i \in N, \forall I \in \I_i, \forall a \in A(I) \\
	& \hat\mu_i^{I,a} \in \argmax_{\hat \mu_i \in \Delta_{\Pi_i(I)}} \left\{  \sum_{z \in Z(I)} p_{\mu,\hat\mu_i}^{\eta, I,a}(z) u_i(z)    \right\} &  \forall i \in N, \forall I \in \I_i, \forall a \in A(I),
	\end{align*}
	where we replaced quantifications over all player $i$'s pure strategies $\hat \mu_i \in \Delta_{\Pi_i(I)}$ with inner maximizations, by introducing auxiliary variables $\hat \mu_i^{I,a}$ for each player $i \in N$, infoset $I \in \I_i$, and action $a \in A(I)$.
	Next, let us notice that, as long as the objective to be maximized  in each inner problem is the sum $\sum_{z \in Z(I)} p_{\mu,\hat\mu_i}^{\eta, I,a}(z) u_i(z) $ (which only contains terms referred to terminal nodes reachable from $I$), strategies $\hat \mu_i \in \Delta_{\Pi_i(I)}$ can be replaced with realization plans $x_i: \Sigma_i \to [0,1]$ such that $x_i(\sigma_i(I))=1$ (\emph{i.e.}, where the probability of reaching infoset $I$ given player $i$'s moves is $1$).
	This holds thanks to the equivalence between mixed strategies and realization plans~\citep{von1996efficient}.
	As a result, for every player $i \in N$, infoset $I \in \I_i$, and action $a \in A(I)$, we can write each inner maximization problem as follows:
	\begin{subequations}\label{prob:primal_inner}
		\begin{align}
		\max & \quad \sum_{z \in Z(I)} \left( \sum_{\substack{\pi_i \in \Pi_i(a)\\\pi_{-i} \in \Pi_{-i}}} \xi^\eta(z,I,\pi) \mu (\pi) \right)  u_i(z) x_i[\sigma_i(z)] \quad \textnormal{s.t.} \label{eq:primal_inner_obj}\\
		& x_i[\sigma_i(I)] = 1 \label{eq:primal_inner_cons1}\\
		& x_i[\sigma_i(J)] = \sum_{a' \in A(J)} x_i[\sigma_i(J)a'] & \forall J \in \I_i: I \preceq J, \forall a' \in A(J)\label{eq:primal_inner_cons2} \\
		& x_i[\sigma_i(J) a'] \geq \eta(a') x_i[\sigma_i(J)] & \forall J \in \I_i: I \preceq J, \forall a' \in A(J) \label{eq:primal_inner_cons3} \\
		& x_i[\sigma_i(I)] \geq 0 \nonumber\\
		& x_i[\sigma_i(J) a'] \geq 0& \forall J \in \I_i: I \preceq J, \forall a' \in A(J) , \nonumber
		\end{align}
	\end{subequations}
	where $x_i[\sigma_i(I)]$ and $x_i[\sigma_i(J) a']$ are variables encoding a player $i$'s realization plan restricted to sequences extending $\sigma_i(I)$ (these are the only variables needed, since Objective~\eqref{eq:primal_inner_obj} does not depend on the realization plan probabilities of other sequences).
	We also notice that the trembles associated with player $i$'s actions at infosets $J \in \I_i : I \preceq J$ (managed by the terms $\xi^{\eta}(z,I,\hat \pi_i)$ in the definition of $p_{\mu,\hat\mu_i}^{\eta, I,a}(z)$) are encoded by Constraints~\eqref{eq:primal_inner_cons3}, which ensure that each action $a' \in A(J)$ is played with probability $\frac{x_i[\sigma_i(J)a']}{x_i[\sigma_i(J)]} \geq \eta(a')$ (given that the denominator is non-null).
	The dual of Problem~\ref{prob:primal_inner} reads as follows:
	\begin{subequations}\label{prob:dual_inner}
		\begin{align}
		\min & \quad v[i,I,a,\varnothing] \quad \textnormal{s.t.} \label{eq:dual_inner_obj}\\
		& v[i,I,a,\varnothing] \geq v[i,I,a,I] + \sum_{a' \in A(I)} \eta(a') w[i,I,a,I,a'] \label{eq:dual_inner_cons1}\\
		& v[i,I,a,J] - w[i,I,a,J,a'] \geq  \sum_{z \in Z^\bot(J,a')} \left( \sum_{\substack{\pi_i \in \Pi_i(a)\\\pi_{-i} \in \Pi_{-i}}} \xi^\eta(z,I,\pi) \mu (\pi) \right)  u_i(z) + \nonumber \\
		& \hspace{.5cm}+ \sum_{K \in \C(J,a')} \left(  v[i,I,a,K] + \sum_{a'' \in A(K)} \eta(a'') w[i,I,a,K,a'']  \right) \hspace{1cm} \forall J \in \I_i: I \preceq J, \forall a' \in A(J)\label{eq:dual_inner_cons2} \\
		& w[i,I,a,J,a'] \geq 0 \hspace{7.6cm} \forall J \in \I_i: I \preceq J, \forall a' \in A(J), \nonumber
		\end{align}
	\end{subequations}
	where $v[i,I,a,\varnothing]$ is the dual variable associated to Constraint~\eqref{eq:primal_inner_cons1}, $v[i,I,a,J]$ for $J \in \I_i : I \preceq J$ are the dual variables associated to Constraints~\eqref{eq:primal_inner_cons2}, and $w[i,I,a,J,a'] $ for $J \in \I_i : I \preceq J$ and $a' \in A(J)$ are the dual variables associated to Constraints~\eqref{eq:primal_inner_cons3}.
	By using the fact that variable $v[i,I,a,\varnothing]$ appears only in Constraint~\eqref{eq:dual_inner_cons1} and by changing sign to variables $w[i,I,a,J,a'] $, we can re-write Problem~\eqref{prob:dual_inner} as follows:
	\begin{subequations}\label{prob:dual_inner_re}
		\begin{align}
		\min & \quad v[i,I,a,I] - \sum_{a' \in A(I)} \eta(a') w[i,I,a,I,a'] \quad \textnormal{s.t.} \label{eq:dual_inner_re_obj}\\
		& v[i,I,a,J] - w[i,I,a,J,a'] \geq \sum_{z \in Z^\bot(J,a')} \left( \sum_{\substack{\pi_i \in \Pi_i(a)\\\pi_{-i} \in \Pi_{-i}}} \xi^\eta(z,I,\pi) \mu (\pi) \right)  u_i(z) + \nonumber \\
		& \hspace{.5cm}+  \sum_{K \in \C(J,a')} \left(  v[i,I,a,K] - \sum_{a'' \in A(K)} \eta(a'') w[i,I,a,K,a'']  \right) \hspace{1cm}   \forall J \in \I_i: I \preceq J, \forall a' \in A(J)\label{eq:dual_inner_ne_cons} \\
		& w[i,I,a,J,a'] \geq 0 \hspace{7.6cm} \forall J \in \I_i: I \preceq J, \forall a' \in A(J). \nonumber
		\end{align}
	\end{subequations}
	Then, we can remove the inner maximization problems by enforcing strong duality, \emph{i.e.}, we add constraints equating Objective~\eqref{eq:primal_inner_obj} and Objective~\eqref{eq:dual_inner_re_obj}.
	Noticing that Objective~\eqref{eq:primal_inner_obj} is equal to $\sum_{z \in Z(I)} p_{\mu,\hat\mu_i^{I,a}}^{\eta, I,a}(z) u_i(z) $, we obtain the following set of linear constraints:
	\begin{align*}
	& \sum_{z \in Z(I)} \left[   \left( \sum_{\substack{\pi_i \in \Pi_i(a) \\ \pi_{-i} \in \Pi_{-i} }} \xi^\eta(z,\pi)  \mu(\pi) \right) u_i(z)   \right] = v[i,I,a,I] - \sum_{a' \in A(I)} \eta(a') w[i,I,a,I,a'] \\
	&  \hspace{12cm}\forall i \in N, \forall I \in \I_i, \forall a \in A(I) \\
	& v[i,I,a,J] - w[i,I,a,J,a'] \geq \sum_{z \in Z^\bot(J,a')} \left( \sum_{\substack{\pi_i \in \Pi_i(a)\\\pi_{-i} \in \Pi_{-i}}} \xi^\eta(z,I,\pi) \mu (\pi) \right)  u_i(z) +  \\
	& \hspace{.5cm}+  \sum_{K \in \C(J,a')} \left(  v[i,I,a,K] - \sum_{a'' \in A(K)} \eta(a'') w[i,I,a,K,a'']  \right) \\
	& \hspace{7.7cm} \forall i \in N, \forall I \in \I_i, \forall a \in A(I), \forall J \in \I_i: I \preceq J, \forall a' \in A(J)\\
	& w[i,I,a,J,a'] \geq 0 \hspace{4.9cm} \forall i \in N, \forall I \in \I_i, \forall a \in A(I),\forall J \in \I_i: I \preceq J, \forall a' \in A(J). 
	\end{align*}
	By introducing variables $u[i,I,a]$ we get to the result.
\end{proof}

\lemmasep*

\begin{proof}
	The proof follows the same line as the proof of Lemma~5~of~\citet{huang2008computing} (its complete version can be found in~\citep{huang2011equilibrium}).
	This is an extension of the CE existence proof by~\citet{hart1989existence} to the case of EFCE.
	It is based on the construction of an auxiliary $2$-player zero-sum EFG, where player $1$ plays first by selecting a strategy profile $\pi \in \Pi$, and player $2$ plays second by choosing an infoset $I \in \I_i$ of some player $i \in N$, an action $a \in A(I)$, and a combinations of actions at following infosets $J \in \I_i : I \preceq J$ (intuitively, player $2$ chooses a trigger agent corresponding to $I$ and $a$, together with a possible trigger agent's behavior).
	It is easy to see that, for our Problem~\ref{prob:dual_ellip_pert}, variables in $\boldsymbol{y}$ have the same meaning as in Lemma~5~of~\citet{huang2008computing}, \emph{i.e.}, they represent valid player $2$'s strategies in the auxiliary game.
	This is because they satisfy the same linear restrictions $B^\top \boldsymbol{y} = \boldsymbol{0}$.
	As a result, the only difference is in the coefficients of the exponentially-many constraints, which, in our case, are defined by the (perturbed) matrix $A_t$, rather than $A$.
	These define the payoffs in the auxiliary game.
	In particular, following steps analogous to those by~\citet{huang2011equilibrium} we can conclude that, in the auxiliary game, player $2$'s expected payment to player $1$ when the latter plays $\pi \in \Pi$ is given by the entry of $A_t^\top \boldsymbol{y}$ corresponding to $\pi$.
	Then, the proof follows the same reasoning as that of~\citet{huang2011equilibrium} to prove the result.
\end{proof}

\end{document}